\documentclass{LMCS}

\def\doi{9(2:06)2013}
\lmcsheading%
{\doi}
{1--31}
{}
{}
{Sep.~10, 2012}
{Jun.~18, 2013}
{}

\usepackage{amsfonts}
\usepackage{amsmath}
\usepackage{amsthm}
\usepackage{amssymb}
\usepackage{subfig}
\usepackage{graphicx}
\usepackage{enumerate,hyperref}

\usepackage{tikz}
\usetikzlibrary{shapes,arrows,decorations.markings}
\tikzstyle{vertex}=[draw,minimum size=0.8cm,inner sep=0cm]
\tikzstyle{even}=[vertex,regular polygon, regular polygon sides=4]
\tikzstyle{odd}=[vertex,regular polygon, regular polygon sides=3]
\tikzstyle{separator}=[vertex,circle,minimum size=0.6cm]
\tikzstyle{node}=[vertex,circle,minimum size=2.4cm]
\tikzset{-<-/.style={decoration={
  markings,
  mark=at position .6 with {\arrow{triangle 45 reversed}}},postaction={decorate}}}
\tikzset{->-/.style={decoration={
  markings,
  mark=at position .8 with {\arrow{triangle 45}}},postaction={decorate}}}

\DeclareMathOperator{\pri}{pri}
\DeclareMathOperator{\io}{Inf}
\DeclareMathOperator{\play}{Play}
\DeclareMathOperator{\maxpri}{MaxPri}
\DeclareMathOperator{\winner}{Winner}
\DeclareMathOperator{\paths}{Paths}
\DeclareMathOperator{\exit}{Exit}
\DeclareMathOperator{\profile}{Profile}
\DeclareMathOperator{\splitgame}{Split}

\DeclareMathOperator{\subtree}{Subtree}
\DeclareMathOperator{\direction}{Direction}
\DeclareMathOperator{\guarded}{Guarded}
\DeclareMathOperator{\minpath}{MinPath}
\DeclareMathOperator{\maxpath}{MaxPath}
\DeclareMathOperator{\update}{\mathsf{Update}}

\DeclareMathOperator{\nextfun}{\mathsf{Next}}
\DeclareMathOperator{\ncnext}{\mathsf{Next_{NC}}}
\DeclareMathOperator{\twnext}{\mathsf{Next_{TW}}}
\DeclareMathOperator{\history}{\mathsf{Hist}}
\DeclareMathOperator{\nchistory}{\mathsf{Hist_{NC}}}
\DeclareMathOperator{\twhistory}{\mathsf{Hist_{TW}}}
\DeclareMathOperator{\point}{Point}
\DeclareMathOperator{\slice}{\mathtt{Slice}}
\DeclareMathOperator{\reduce}{\mathtt{Reduce}}
\DeclareMathOperator{\sr}{\mathtt{SR}}
\DeclareMathOperator{\first}{First}

\DeclareMathOperator{\simulate}{\mathsf{Simulate}}

\newcommand{\nats}{\mathbb N}

\newcommand{\acc}{\texttt{accept}}
\newcommand{\rej}{\texttt{reject}}
\newcommand{\follow}{\mathtt{follow}}

\begin{document}

\title[Time and Parallelizability Results for Parity Games with Bounded Width]{Time and Parallelizability Results for Parity Games with Bounded Tree and DAG Width}
\author[J.~Fearnley]{John Fearnley}
\address{Department of Computer Science, University of Liverpool, Liverpool, UK}
\email{\{john.fearnley,\,sven.schewe\}@liverpool.ac.uk}

\author[S.~Schewe]{Sven Schewe}
\address{\vskip-6 pt}

 \ACMCCS{[{\bf
      Theory of computation}]: Logic---Verification by model checking;
  Design and analysis of algorithms---Parameterized complexity and
  exact algorithms}

\begin{abstract}
Parity games are a much researched class of games in NP~$\cap$ CoNP that are not
known to be in P. Consequently, researchers have considered specialised
algorithms for the case where certain graph parameters are small. In this paper,
we study parity games on graphs with bounded treewidth, and graphs with bounded
DAG width. We show that parity games with bounded DAG width can be solved in
$O(n^{k+3} \cdot k^{k + 2} \cdot (d + 1)^{3k + 2})$ time, where~$n$,~$k$,
and~$d$ are the size, treewidth, and number of priorities in the parity game.
This is an improvement over the previous best algorithm, given by Berwanger et
al., which runs in $n^{O(k^2)}$ time. We also show that, if a tree
decomposition is provided, then parity games with bounded treewidth can be
solved in $O(n \cdot k^{k + 5} \cdot (d + 1)^{3k + 5})$ time. This improves over
previous best algorithm, given by Obdr\v{z}\'alek, which runs in $O(n \cdot
d^{2(k+1)^2})$ time. Our techniques can also be adapted to show that the problem
of solving parity games with bounded treewidth lies in the complexity class
NC$^2$, which is the class of problems that can be efficiently parallelized.
This is in stark contrast to the general parity game problem, which is known to
be P-hard, and thus unlikely to be contained in NC.
\end{abstract}

\maketitle

\section{Introduction}
A parity game is a two player game that is played on a finite directed graph.
The problem of solving a parity game is known to lie in
NP~$\cap$~CoNP~\cite{McNaughton/93/Games}, and a sub-exponential time algorithm
is known for the problem~\cite{JPZ08}. However, despite much effort, this
problem is not known to be in P. Due to the apparent difficulty of solving
parity games, recent work has considered special cases, where the input is
restricted in some way. In particular, people have studied parity games where
the input graph is restricted by a graph parameter. For example, parity games
have been shown to admit polynomial time algorithms whenever the input graph has
bounded treewidth~\cite{Obdrzalek/03/TreeWidth}, DAG width
\cite{Berwanger+all/06/ParityDAG}, clique width~\cite{O07}, Kelly
width~\cite{HK08}, or entanglement \cite{Berwanger+Graedel/04/entanglement}. In
this paper we study the parity game problem for graphs of bounded treewidth and
graphs of bounded DAG width.

Parity games are motivated by applications in model checking. The problem of
solving a parity game is polynomial-time equivalent to the modal $\mu$-calculus
model checking problem~\cite{Emerson+all/93/mu,stirling95}. In model checking,
we typically want to check whether a large system satisfies a much smaller
formula. It has been shown that many practical systems have bounded treewidth.
For example, it has been shown that the control flow graphs of goto free Pascal
programs have treewidth at most~$3$, and that the control flow graphs of goto
free C programs have treewidth at most~$6$~\cite{Mikkel98}. The same paper also
shows that tree decompositions, which are costly to compute in general, can
be generated in linear time with small constants for these control flow graphs.
Moreover, Obdr\v{z}\'alek  has shown that, if the input system has treewidth
$k$, and if the $\mu$-calculus formula has $m$ sub-formulas, then the modal
$\mu$-calculus model checking problem can be solved by determining the winner of
a parity game that has treewidth at most $k \cdot
m$~\cite{Obdrzalek/03/TreeWidth}. Since~$m$ is usually much smaller than the
size of the system, we have a strong motivation for solving parity games with
bounded treewidth.

Parity games with bounded treewidth were first studied by
Obdr\v{z}\'alek~\cite{Obdrzalek/03/TreeWidth}. He gave an algorithm that runs in
$O(n \cdot k^2 \cdot d^{2(k+1)^2})$ time, while using $d^{O(k^2)}$ space, where
$d$ is the number of priorities in the parity game, and~$k$ is the treewidth of
the game. This result shows that, if the treewidth of a parity game is bounded,
then there is a polynomial time algorithm for solving the game. However, since
the degree of this polynomial depends on $k^2$, this algorithm is not at all
practical. For example, if we wanted to solve the model checking problem for the
control flow graph of a C program with treewidth $6$, and a $\mu$-calculus
formula with a single sub-formula, the running time of the algorithm will
already be $O(n \cdot d^{98})$.

DAG width is a generalisation of treewidth to directed graphs. It was defined by
Berwanger et al.\ ~\cite{Berwanger+all/06/ParityDAG,BDHKO12}, and independently
by Obdr\v{z}\'alek~\cite{O06}. Since parity games are played on
directed graphs, it is natural to ask whether Obdr\v{z}\'alek's algorithm can be
generalised to parity games with bounded DAG width. Berwanger et al.\
~\cite{Berwanger+all/06/ParityDAG,BDHKO12} showed that the algorithm can indeed
be generalised. Their work is a generalisation of Obdr\v{z}\'alek's techniques
to parity games with bounded DAG width, but since every tree decomposition can
also be viewed as a DAG decomposition, their algorithm can also be applied to
graphs with bounded treewidth. In comparison to Obdr\v{z}\'alek's original
algorithm, they achieve an improved space complexity of $d^{O(k)}$. However,
some operations of their algorithm still take $d^{O(k^2)}$ time, and thus their
algorithm is no faster than the original.

More recently, an algorithm has been proposed for parity games with ``medium''
tree width~\cite{Fearnley+Lachish/11/treewidth}. This algorithm runs in time
$n^{O(k \cdot \log n)}$, and it is therefore better than Obdr\v{z}\'alek's
algorithm whenever $k \in \omega(\log n)$. On the other hand, this algorithm
does not provide an improvement for parity games with small treewidth.

\subsection{Our contribution.} 

In this paper, we take a step towards providing a practical algorithm for these
problems. We show how parity games with bounded treewidth and DAG width can be
solved solved in logarithmic space on an alternating Turing machine. We then use
different variants of this technique to show three different results. Our first
result is for parity games with bounded DAG width. We show that these games can
be solved by an alternating Turing machine while using $(k + 3) \cdot \log n +
(k + 2) \cdot \log(k) + (3k + 2) \cdot \log(d + 1)$ bits of storage. This
then implies that there is a deterministic algorithm that runs in $O(n^{k+3}
\cdot k^{k + 2} \cdot (d + 1)^{3k + 2})$ time. Since the exponent of $n$ is
$k+3$, rather than $k^2$, this result improves over the $n^{O(k^2)}$ algorithm
given by Berwanger et al.

We then turn our attention to parity games with bounded treewidth. Using similar
techniques, we are able to obtain an algorithm that solves parity games with
bounded treewidth in 
$O(n \cdot (k+1)^{k+5} \cdot (d+1)^{3k + 5})$ 
time on a
deterministic Turing machine. Again, since the exponent of $d$ does not depend
on $k^2$, this is an improvement over the original algorithm of Obdr\v{z}\'alek,
which runs in $O(n \cdot k^2 \cdot d^{2(k+1)^2})$ time.

For the treewidth result, we must pay attention to the complexity of computing a
tree decomposition. In the DAG width result, this did not concern use, because a
DAG decomposition can be computed in $O(n^{k+2})$ time, and this is therefore
absorbed by the running time of the algorithm. In the treewidth case, however,
things are more complicated. From a theoretical point of view, we could apply
the algorithm of Bodlaender~\cite{bodlaender96}, which runs in $O(n \cdot f(k))$
time. However, the function $f$ lies in $2^{O(k^3)}$, and this has the potential
to dwarf any improvement that our algorithm provides. As we have
mentioned, in some practical cases, a tree decomposition can be computed
cheaply. Otherwise, we suggest that the approximation algorithm of
Amir~\cite{amir01} should be used to find a $4.5$-approximate tree
decomposition. Our algorithm runs in  $O(n \cdot k^{4.5k + 5} \cdot (d +
1)^{13.5k + 5})$ time when a $4.5$-approximate tree decomposition is used. Note
that Obdr\v{z}\'alek's algorithm also requires a tree decomposition as input.
Thus, in this case, the gap between the two algorithms is even wider, as using a
$4.5$-approximate tree decomposition causes Obdr\v{z}\'alek's algorithm to run
in $O(n \cdot d^{40.5k^2 + 18k + 2})$ time.

Finally, we are able to adapt these techniques to show a parallelizability
result for parity games with bounded treewidth. We are able to provide an
alternating Turing machine that solves the problem in $O(k^2 \cdot (\log n)^2)$
time while using $O(k \cdot \log n)$ space. This version of the algorithm does
not require a precomputed tree decomposition. Hence, using standard results in
complexity theory~\cite{Allender}, we have that the problem lies in the
complexity class NC$^2 \subseteq$ NC, which is the class of problems that can be
efficiently parallelized. 

This result can be seen in stark contrast to the complexity of parity games on
general graphs: parity games are known to be P-hard by a reduction from
reachability games, and P-hardness is considered to be strong evidence that an
efficient parallel algorithm does not exist. Our result here shows that, while
we may be unable to efficiently parallelize the $\mu$-calculus model checking
problem itself, we can expect to find efficient parallel algorithms for the
model checking problems that appear in practice.

\section{Preliminaries}

\subsection{Parity games}

A parity game is a tuple $(V, V_\text{0}, V_\text{1}, E, \pri)$,
where~$V$ is a set of vertices and~$E$ is a set of edges, which together form a
finite directed graph. The sets~$V_\text{0}$ and $V_\text{1}$ partition~$V$ into
vertices belonging to player Even and player Odd, respectively. The function
$\pri : V \rightarrow D$ assigns a \emph{priority} to each vertex from the set
of priorities~$D \subseteq \nats$. It is required that the game does not contain
any dead ends: for each vertex $v \in V$ there must exist an edge $(v,
u) \in E$.

We define the \emph{significance} ordering~$\prec$ over $D$. This ordering
represents how attractive each priority is to player Even. For two priorities
$a, b \in \nats$, we have $a \prec b$ if one of the following conditions holds:
(1) $a$ is odd and $b$ is even, (2) $a$ and $b$ are both even and $a < b$, or
(3) $a$ and $b$ are both odd and $a > b$. We say that $a \preceq b$ if either $a
\prec b$ or $a = b$.

At the beginning of the game, a token is placed on a starting vertex~$v_0$. In
each step, the owner of the vertex that holds the token must choose one outgoing
edge from that vertex and move the token along it. In this fashion, the two
players form an infinite path $\pi = \langle v_0, v_1, v_2, \dots \rangle$,
where~$(v_i, v_{i+1}) \in E$ for every~$i \in \nats$. To determine the winner of
the game, we consider the set of priorities that occur \emph{infinitely often}
along the path. This is defined to be: $\io(\pi) = \{ d \in \nats \; : \;
\text{For all } j \in \nats \text{ there is an } i > j \text{ such that
}\pri(v_i) = d \}.$ Player Even wins the game if the highest priority occurring
infinitely often is even, and player Odd wins the game if it is odd. In other
words, player Even wins the game if and only if $\max(\io(\pi))$ is even.

A positional strategy for Even is a function that chooses one outgoing edge for
every vertex in $V_\text{0}$. A strategy is denoted by $\sigma :
V_\text{0} \rightarrow V$, with the condition that $(v, \sigma(v)) \in E$
for every Even vertex~$v$. Positional strategies for player Odd are defined
analogously. The sets of positional strategies for Even and Odd are denoted by
$\Sigma_\text{0}$ and $\Sigma_\text{1}$, respectively.  Given two positional
strategies~$\sigma$ and~$\tau$, for Even and Odd, respectively, and a starting
vertex~$v_0$, there is a unique path $\langle v_0, v_1, v_2 \dots \rangle$,
where~$v_{i+1} = \sigma(v_i)$ if~$v_i$ is owned by Even, and~$v_{i+1} =
\tau(v_i)$ if~$v_i$ is owned by Odd. This path is known as the \emph{play}
induced by the two strategies~$\sigma$ and~$\tau$, and will be denoted by
$\play(v_0, \sigma, \tau)$.

For each $\sigma \in \Sigma_0$, we define $G \restriction \sigma$ to be
the modification of $G$ where Even is forced to play~$\sigma$. That
is, an edge $(v, u) \in E$ is included in $G \restriction \sigma$ if either $v
\in V_1$, or $v \in V_0$ and $\sigma(v) = u$. We define $G \restriction \tau$
for all $\tau \in \Sigma_1$ analogously.

An infinite path $\langle v_0, v_1, \dots \rangle$ is said to be
\emph{consistent} with an Even strategy $\sigma \in \Sigma_\text{0}$ if $v_{i+1}
= \sigma(v_i)$ for every $i$ such that $v_i \in V_\text{0}$. If~$\sigma \in
\Sigma_\text{0}$ is a strategy for Even, and~$v_0$ is a starting vertex, then we
define $\paths(v_0, \sigma)$ to give every path starting
at $v_0$ that is consistent with $\sigma$. An Even strategy $\sigma \in
\Sigma_0$ is called a \emph{winning strategy} for a vertex $v_0 \in V$ if
$\max(\io(\pi))$ is even for all $\pi \in \paths_\text{0}(v_0, \sigma)$.
The strategy~$\sigma$ is said to be winning for a set of vertices $W \subseteq
V$ if it is winning for all~$v \in W$. Winning strategies for player
Odd are defined analogously. 

A game is said to be \emph{positionally determined} if one of the two players
always has a positional winning strategy. We now give a fundamental theorem,
which states that parity games are positionally determined.

\begin{thm}[\cite{Emerson+Jutla/91/Memoryless,mostowski91}]
\label{paritydeterminacy}
In every parity game, the set of vertices~$V$ can be partitioned into
\emph{winning sets} $(W_\text{0}, W_\text{1})$, where Even has a positional
winning strategy for $W_\text{0}$, and Odd has a positional winning strategy for
$W_\text{1}$.
\end{thm}

In this paper we study the following computational problem for parity games:
given a starting vertex $s$, determine whether $s \in W_0$ or $s \in W_1$.

\subsection{Treewidth.}

Treewidth originated from the work of
Robertson and Seymour~\cite{RS84}. Treewidth is a complexity measure for undirected graphs.
Thus, to define the treewidth of a parity game, we will use the treewidth of the
undirected graph that is obtained when the orientation of the edges is ignored.
We begin by defining tree decompositions.

\begin{defi}[Tree Decomposition]
\label{def:treedecomp}
For each game~$G = (V, V_\text{0}, V_\text{1}, E, \pri)$, the pair $(T, X)$,
where $T = (I, J)$ is an undirected tree and $X = \{X_i \; : \; i \in I\}$ is a
family of subsets of~$V$, is a tree decomposition of~$G$ if all of the following
hold:
\begin{enumerate}[(1)]
\item $\bigcup_{i \in I} X_i = V$.
\item For every $(v, u) \in E$ there is an $i \in I$ such that $v \in X_i$
and $u \in X_i$.
\item For every $i, j \in I$, if $k \in I$ is on the unique path from $i$ to $j$
in $T$, then $X_i \cap X_j \subseteq X_k$.
\end{enumerate}
\end{defi}

\noindent The \emph{width} of a tree decomposition $(T, X)$ is $\max\{|X_i| \; :
\; i \in I\}$. The \emph{treewidth} of a game~$G$ is the smallest width of a
tree decomposition of~$G$. Note that this is a slightly unusual definition,
because the width of a tree decomposition is usually defined to be $\max\{|X_i|
- 1\; : \; i \in I\}$. However, we use our definition in order to keep the
definitions the same between treewidth and DAG width.

Let $(T = (I, J), X)$ be a tree decomposition for a parity game $G = (V,
V_\text{0}, V_\text{1}, E, \pri)$. Let $i \in I$ be a node in the tree
decomposition, and let $v \in V$ be a vertex in the parity game with $v \notin
X_i$. Let $k \in I$ be the closest node to $i$ in $T$ such that $x \in X_k$.
Since $T$ is a tree decomposition, the node $k$ is well defined. Let $(i, j) \in
J$ be the first edge on the path from $i$ to $k$. We define $\direction(X_i, v)$
to be the function that outputs the node $j$.

There are multiple approaches for computing tree decompositions. From a
theoretical point of view, the best known algorithm is the algorithm of
Bodlaender~\cite{bodlaender96}. If the treewidth is bounded, then this is a
linear time algorithm: it runs in $O(n \cdot f(k))$ time. However, the constant
factor hidden by the function $f(k)$ is in the order of $2^{O(k^3)}$, which
makes the algorithm impractical. For a more practical approach, we can apply the
algorithm of Amir~\cite{amir01} to approximate the treewidth of the graph. This
algorithm takes a graph~$G$ and an integer~$k$, and in $O(2^{3k} \cdot n^2 \cdot k^{3/2})$
time either finds a tree decomposition of width at most~$4.5k$ for~$G$, or
reports that the tree-width of~$G$ is larger than~$k$.

In this paper, we will assume that the size of the tree decomposition is linear
in the size of the parity game. More precisely, we assume that, if $(T = (I, J),
X)$ is a tree decomposition of a parity game $(V, V_\text{0}, V_\text{1}, E,
\pri)$, then we have $|I| \le |V|$. It has been shown that every tree
decomposition can be modified, in polynomial time, to satisfy this
property~\cite[Lemma 2.2]{bodlaender96}\footnote{Specifically, we refer to Lemma
2.2 in the SIAM Journal on Computing version of this paper.}. Therefore we can
make this assumption without loss of generality.

\subsection{DAG width.}
\label{sec:dagwidth}

As opposed to treewidth, which is a measure for undirected graphs, DAG
width~\cite{BDHKO12} is a measure for directed graphs. A directed graph $G = (V,
E)$ is a DAG if it contains no directed cycles. If $G$ is a DAG, then we define
$\sqsubseteq_D$ to be the reflexive transitive closure of the edge relation of~$G$.
A \emph{source} in the DAG is a vertex $v \in V$ that is minimal in the
$\sqsubseteq_D$ ordering, and a \emph{sink} in the DAG is a vertex that is maximal
in the $\sqsubseteq_D$ ordering. Furthermore, given two sets $U, W \subseteq V$, we
say that $W$ \emph{guards} $U$ if, for every edge $(v, u) \in E$, where $v \in
U$, we have $u \in W \cup U$. In other words, $W$ guards $U$ if the only way to
leave $U$ is to pass through a vertex in $W$. We can now define a DAG
decomposition. 

\begin{defi}[DAG Decomposition]
\label{def:dagwidth}
Let $G = (V, V_\text{0}, V_\text{1}, E, \pri)$ be a parity game. A DAG
decomposition of $G$ is a pair $(\mathcal{D} = (I, J), X)$, where $\mathcal{D}$ is a DAG, and
$X = \{X_i \; : \; i \in I\}$ is a family of subsets of~$V$, which satisfies the
following conditions:
\begin{enumerate}[(1)]
\item $\bigcup_{i \in I} X_i = V$.
\item For every edge $(i, j) \in J$, the set $X_i \cap X_j$ guards $(\bigcup_{j
\sqsubseteq_D k} X_k) \setminus X_i$.
\item For every $i, j, k \in I$, if $i \sqsubseteq_D k \sqsubseteq_D j$, then $X_i \cap X_j \subseteq X_k$.
\end{enumerate}
\end{defi}

\noindent The width of a DAG decomposition $(D, X)$ is $\max\{|X_i| \; : \; i
\in I\}$. The \emph{DAG width} of a game~$G$ is the smallest width of a DAG
decomposition of~$G$. In accordance with the second condition in
Definition~\ref{def:dagwidth}, for each $X_i \in X$ we define:
\begin{equation*}
\guarded(X_i) =
(\bigcup_{(i, j) \in J} (\bigcup_{j \sqsubseteq_D k} X_k)) \setminus X_i.
\end{equation*}

We can also define $\direction$ for DAG decompositions. Suppose that
$(\mathcal{D} = (I,
J), X)$ is a DAG decomposition of a parity game $G = (V, V_\text{0}, V_\text{1},
E, \pri)$. Let $i \in I$ be a node in the DAG decomposition, and let $v \in V$
be a vertex in the parity game with $v \in \guarded(X_i)$. From the properties
of a DAG decomposition, there must be at least one $j \in I$ such that $(i, j)
\in J$ and either $v \in X_j$ or $v \in \guarded(X_j)$. We define
$\direction(X_i, v)$ to arbitrarily select a node $j \in I$ that satisfies this
property.

The only algorithm for computing DAG decompositions was given by
Berwanger et. al.\ ~\cite{BDHKO12}. They showed that, if a graph has DAG width
$k$, then a DAG decomposition can be computed in $O(n^{k+2})$ time. In the case
of DAG width, we cannot assume that the size of the DAG decomposition is linear
in the size of the graph. In fact, the best known upper bound on the number of
nodes in a DAG decomposition is $n^{k+1}$, and the number of edges is $n^{k+2}$,
where $n$ is the number of vertices in the graph, and $k$ is the width of the
DAG decomposition~\cite[Proof of Theorem 16]{BDHKO12}. 



\section{Strategy Profiles}

In this section we define strategy profiles, which are data structures that
allow a player to give a compact representation of the relevant properties of
their strategy. 
Our algorithms in Sections~\ref{sec:time},~\ref{sec:complexity},
and~\ref{sec:space} will use strategy profiles to allow the players to declare
their strategies in a small amount of space. Throughout this section we will
assume that there is a \emph{starting vertex} $s \in V$ and a set of \emph{final
vertices} $F \subseteq V$. Let $\sigma \in \Sigma_0$ be a strategy for Even. The
strategy profile of~$\sigma$ describes the outcome of a modified parity game
that starts at $s$, terminates whenever a vertex $u \in F$ is encountered, and
in which Even is restricted to only play $\sigma$.

For each $u \in F$, we define $\paths(\sigma, s, F, u)$ to be the set of paths
from~$s$ to~$u$ that are consistent with~$\sigma$ and that do not visit a vertex
in~$F$. More formally, $\paths(\sigma, s, F, u)$ contains every path of the form
$\langle v_0, v_1, v_2, \dots v_k \rangle$ in $G \restriction \sigma$, for which
both of the following conditions hold: 
\begin{iteMize}{$\bullet$}
\item the vertex $v_0 = s$ and the vertex $v_k = u$, and 
\item for all $i$ in the range $0 \le i \le k-1$ we have $v_i \notin F$.
\end{iteMize}
For each strategy $\tau \in \Sigma_1$, we define the functions $\paths(\tau, s,
F, u)$ analogously.

Recall that $\preceq$ is the significance ordering over priorities. For each $u \in F$, the function $\exit(\sigma, s, F, u)$, gives the best
possible priority, according to $\preceq$, that Odd can visit when Even plays
$\sigma$ and Odd chooses to move to~$u$. This function either gives a priority
$p \in D$, or, if Odd can never move to~$u$ when Even plays $\sigma$, the
function gives a special symbol $-$, which stands for ``unreachable''. We will
also define this function for Odd strategies $\tau \in \Sigma_1$. Formally, for
every finite path $\pi = \langle v_1, v_2, \dots, v_k \rangle$, we define
$\maxpri(\pi) = \max \{\pri(v_i) \; : \; 1 \le i \le k\}$. Furthermore, we
define, for $\sigma \in \Sigma_0$ and $\tau \in \Sigma_1$:
\begin{align*}
\minpath(\sigma, s, F, u) &= \min_{\preceq} \{\maxpri(\pi) : \pi \in
\paths(\sigma,
s, F, u)\}, \\
\maxpath(\tau, s, F, u) &= \max_{\preceq} \{\maxpri(\pi) : \pi \in \paths(\tau,
s, F, u)\}.
\end{align*}
For every strategy $\chi \in \Sigma_0 \cup \Sigma_1$ and every $u
\in F$ we define:
\begin{equation*}
\exit(\chi, s, F, u) = \begin{cases}
- & \text{if $\paths(\chi, s, F, u) = \emptyset$,} \\
\minpath(\chi, s, F, u)&
\text{if $\paths(\chi, s, F, u) \ne \emptyset$ and $\chi \in \Sigma_0$,} \\
\maxpath(\chi, s, F, u)&
\text{if $\paths(\chi, s, F, u) \ne \emptyset$ and $\chi \in \Sigma_1$.}
\end{cases}
\end{equation*}

We can now define the strategy profile for each strategy $\chi \in \Sigma_0 \cup
\Sigma_1$. We define $\profile(\chi, s, F)$ to be a function $F \rightarrow D
\cup \{-\}$ such that $\profile(\chi, s, F)(u) = \exit(\chi, s, F, u)$ for each
$u \in F$. 

\section{Outline}

In this section we give an outline of the rest of the paper. We start by
describing simulated parity games, which are the foundation upon which all our
results are based. Then we describe how this simulation game can be used to
prove our three results.

\subsection{Simulated Parity Games}

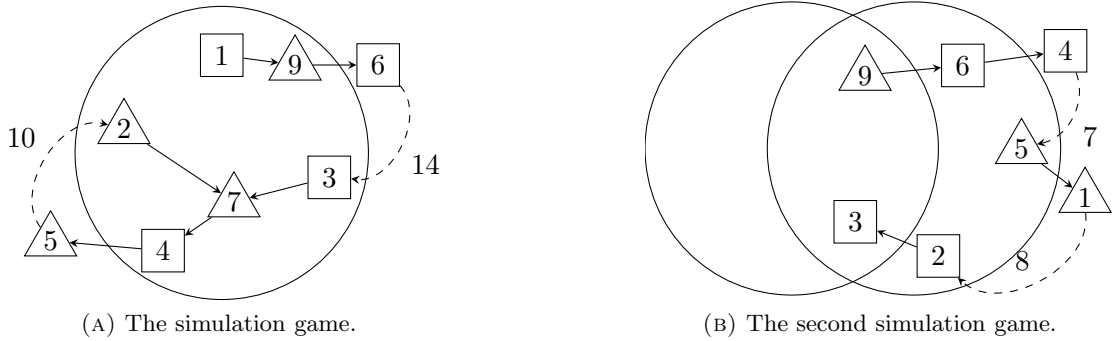
\begin{figure}
\subfloat[The simulation game.]
{
	\label{fig:sim}
	\begin{tikzpicture}[>=stealth,scale=0.65]
	\draw (0, 0) circle (3cm);
	\node (2) at (-2,0.5) [odd] {2};
	\node (5) at (-3.5,-1.8) [odd] {5}
		edge[->, bend left=70,dashed] node[auto] {10} (2);
	\node (4) at (-1.2,-2) [even] {4}
		edge[->] (5);
	\node (1p) at (0.25,-1) [odd] {7}
		edge[->] (4)
		edge[<-] (2);
	\node (3) at (2.2,-0.5) [even] {3}
		edge[->] (1p);
	\node (6) at (3.2,1.8) [even] {6}
		edge[->, bend left = 70,dashed] node[auto] {14} (3);
	\node (9) at (1.5,1.8)  [odd] {9}
		edge[->] (6);
	\node (e1) at (0,2)  [even] {1}
		edge[->] (9);
	\end{tikzpicture}
}
\hfill
\subfloat[The second simulation game.]
{
	\label{fig:verification}
	\begin{tikzpicture}[>=stealth,scale=0.65]
	\clip (-3.0,3.1) rectangle (6.6,-3.1);
	\draw (0, 0) circle (3cm);
	\draw (2.5, 0) circle (3cm);
	\node (3) at (1.3,-1.5) [even] {3};
	\node (2) at (3,-2.2) [even] {2}
		edge[->] (3);
	\node (1) at (6.0, -1.0) [odd] {1}
		edge[->, bend left=70,dashed] node[auto,swap] {8} (2);
	\node (5) at (4.7,0) [odd] {5}
		edge[->] (1);
	\node (4) at (5.6,2) [even] {4}
		edge[->, bend left=50,dashed] node[auto] {7} (5);
	\node (6) at (3.5,1.7) [even] {6}
		edge[->] (4);
	\node (9) at (1.5,1.5) [odd] {9}
		edge[->] (6);
	\end{tikzpicture}
}
\caption{Example runs of the simulation game.}
\label{fig:example}
\end{figure}

Suppose that we have a DAG or tree decomposition of width~$k$ for our parity
game. Suppose further that we want to determine the winner of some vertex $s \in
V$ in the parity game. Our approach is to find some set of vertices $S \subseteq
V$, where $S = X_i$ for some node $i$ in our decomposition, such that $s \in S$.
We then play a \emph{simulation game} on $S$, which simulates the whole parity
game using only the vertices in $S$.

The general idea behind the simulation game is shown in
Figure~\textsc{\ref{fig:sim}}. The large circle depicts the decomposition
node~$S$, the boxes represent Even vertices, and the triangles represent Odd
vertices. Since no two vertices share the same priority in this example, we will
use the priorities to identify the vertices. As long as both players choose to
remain in~$S$, the parity game is played as normal. However, whenever one of the
two players chooses to move to a vertex $v$ with~$v \notin S$ we simulate the
parity game using a strategy profile: Even is required to declare a strategy in
the form of a strategy profile $P$ for~$v$ and~$S$. Odd then picks some vertex
$u \in S$, and moves there with priority $P(u)$\footnote{The assignment of
players here is arbitrary: none of our results would change if we had Odd
propose a strategy profile, and Even pick a vertex in response.}. In the
diagram, the dashed edges represent these simulated decisions. For example, when
the play moved to the vertex~$6$, Even gave a strategy profile~$P$ with $P(3) =
14$, and Odd decided to move to the vertex~$3$. Together, the simulated and real
edges will eventually form a cycle, and the winner of the game is the winner of
this cycle. In our example, Even wins the game because the largest priority on
the cycle is $10$.

If Even always gives strategy profiles that correspond to some strategy $\sigma
\in \Sigma_0$, then the outcome of our simulated game will match the outcome
that would occur in the real parity game. On the other hand, it is possible that
Even could lie by giving a strategy profile~$P$ for which he has no strategy. To
deal with this, Odd is allowed to reject~$P$, which causes the two players to
move to a second simulation game, where Even is required to prove that he
does have a strategy for~$P$. 

Suppose that Odd rejected Even's first strategy profile in the game shown in
Figure~\textsc{\ref{fig:sim}}. Hence, Even must show that his strategy
profile~$P$, that contains $P(3) = 14$, is correct. To do this, we select a
second node $S' \subseteq V$ from the decomposition, and play a simulation game
on $S'$. An example run of this game is shown in
Figure~\textsc{\ref{fig:verification}}. The left circle represents $S$, and the
right circle represents $S'$. The game proceeds as before, by simulating a
parity game on $S'$. However, we add the additional constraint that, if a vertex
$u \in S$ is visited, then the game ends: Even wins the game if the largest
priority~$p$ seen on the path to~$u$ has $p \succeq P(u)$, and Odd wins
otherwise. In this example, Even loses the game because the largest priority
seen during the path is~$8$, and Even's strategy profile in
Figure~\textsc{\ref{fig:sim}} claimed that the largest priority $p$ should
satisfy $p \succeq 14$. We also use rejections to deal with the case where
Even's strategy profile never returns to $S$. If Even gives a strategy profile
$P$ with $P(u) = -$ for every $u \in S$, then Odd will reject $P$, and the game
will continue on $S'$ as before.

The simulation game will be formally defined in Section~\ref{sec:game}. We will
then go on to prove three different results using three distinct versions of the
simulation game. These versions differ in the way that the set~$S'$ is selected
whenever Odd rejects Even's strategy profile.

\subsection{The Time Complexity Results}
\label{sec:outline1}

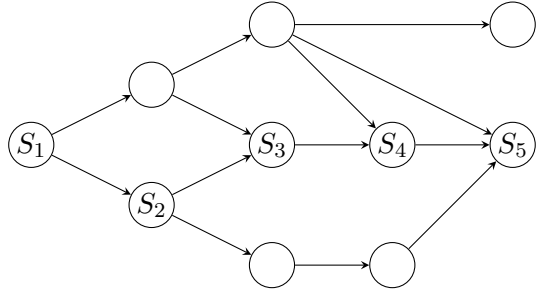
\begin{figure}
\begin{center}
\begin{tikzpicture}[>=stealth,scale=0.8]
\node (1) at (0, 0) [separator] {$S_1$};
\node (2) at (2, 1) [separator] {}
	edge[<-] (1);
\node (3) at (2, -1) [separator] {$S_2$}
	edge[<-] (1);
\node (4) at (4, 2) [separator] {}
	edge[<-] (2);
\node (5) at (4, 0) [separator] {$S_3$}
	edge[<-] (2)
	edge[<-] (3);
\node (6) at (4, -2) [separator] {}
	edge[<-] (3);
\node (7) at (6, -2) [separator] {}
	edge[<-] (6);
\node (8) at (6, 0) [separator] {$S_4$}
	edge[<-] (5)
	edge[<-] (4);
\node (9) at (8, 0) [separator] {$S_5$}
	edge[<-] (4)
	edge[<-] (8)
	edge[<-] (7);
\node (10) at (8, 2) [separator] {}
	edge[<-] (4);
\end{tikzpicture}
\end{center}
\caption{An example run of the simulation game for the time complexity result.}
\label{fig:timec}
\end{figure}

In our first result, to be shown in Section~\ref{sec:time}, we will show
that parity games with DAG width $k$ can be solved in $O(|V|^{k+3} \cdot k^{k +
2} \cdot (|D| + 1)^{3k + 2})$ time. Figure~\ref{fig:timec} gives an example run
of the simulation game that will be used in this result. The figure shows the
DAG decomposition of the parity game, and each circle depicts one of the nodes
in the decomposition. At the start of the game, we find a source node~$i$ of the
DAG decomposition, and we play a simulation game on the set of vertices $X_i$.
This node is shown as $S_1$ in the figure.

Recall, from Figure~\textsc{\ref{fig:example}}, that if the game on $S_1$ ends
with Odd rejecting Even's strategy profile, then we must pick a new set of
vertices, and play a second simulation game. The rule for picking this new set
of vertices will make use of the DAG decomposition. Suppose that the game on
$S_1$ ends when Odd rejects Even's strategy profile for $v$ and $S_1$. The next
set of vertices is chosen to be~$X_i$, where $i = \direction(S_1, v)$. This set
is shown as~$S_2$ in the diagram.

Thus, as we can see in the diagram, as Odd keeps rejecting Even's strategy
profiles, we walk along a path in the DAG decomposition. This means that the
game must eventually end. To see this, note that when we play a simulation game
on~$S_5$, the properties of a DAG decomposition ensure that we either form a
cycle in~$S_5$, or that we visit some vertex in~$S_4$, because~$S_4$ is a guard
of~$S_5$. Recall that the simulation game ends whenever we move back to a set
that we have already seen. Thus, when we play a simulation game on a sink in the
DAG decomposition, there cannot be any simulated moves, and the game will either
end in a cycle on the vertices in $S_5$, or when one of the two players moves to
a vertex in $S_4$.

Our plan is to implement the simulation game on an alternating Turing machine.
We will use the non-deterministic and universal states in the machine to
implement the moves of the two players. Our goal is to show that this
implementation uses at most $O(k \cdot \log |V|)$ space,  which would then
immediately imply our desired result.

There is one observation about DAG decompositions that is important for
obtaining the $O(k \cdot \log |V|)$ alternating space bound. Recall from
Figure~\textsc{\ref{fig:verification}} that if Odd rejects a strategy profile of Even,
then we must remember the strategy profile so that we can decide the winner in
subsequent simulation games. Our observation is that, if we are playing a
simulation game using a DAG decomposition, then we only ever have to
remember one previous strategy profile. For example, suppose that we are playing
the simulation game on $S_3$ in Figure~\ref{fig:timec}. Since $S_3 \subseteq
\guarded(S_2)$, we know that we cannot reach a vertex in $S_1$ without passing
through a vertex in $S_2$. However, the simulation game on $S_3$ ends
immediately when a in $S_2$ is visited. Hence, we can forget the strategy
profile on $S_1$. This observation is crucial for showing the $O(k \cdot \log
|V|)$ alternating space bound.

Our second result, which will be shown in Section~\ref{sec:complexity}, uses the
same techniques, but applies them to parity games with bounded treewidth. We
will show how the amount of space used by the alternating Turing machine can be
significantly reduced for the treewidth case, and from this we derive a $O(|V|
\cdot (k+1)^{k+5} \cdot (|D|+1)^{3k + 5})$ time algorithm for parity games with
bounded treewidth.

\subsection{The Parallelizability Result}
\label{sec:outline2}

In our third result, which will be proved in Section~\ref{sec:space}, we 
show that the problem of solving  parity games with bounded tree-width lies in
the complexity class NC$^2$. To do this, we will construct a rather different
version of the simulation game, which can be solved by an alternating Turing
machine in $O(k^2 \cdot (\log |V|)^2)$ time and $O(k \cdot \log |V|)$ space.
This then immediately implies that our problem lies in the class
NC$^2$~\cite[Theorem 22.15]{Allender}.

In the first and second results, we were able to compute a
DAG decomposition or tree decomposition, and then chose the sets $S$ to be nodes
in this decomposition. However, we do not take this approach here, because we do
not know of a way to compute a tree decomposition in $O(k \cdot \log |V|)$ space
on an alternating Turing machine. Instead, we will allow player Odd to chose
these sets.

However, giving Odd this freedom comes at a price. In the first result, we only
ever had to remember the last strategy profile that was rejected by Odd, and
this was critical for showing the required space bounds for the alternating
Turing machine. In this result, since Odd could potentially select any set, we may have to
remember all of the previous strategy profiles to ensure that the simulation
game eventually terminates. But, if we remember too many strategy profiles, then
we will be unable to show the required space bound for the alternating Turing
machine.

We resolve this by showing the following property: if the parity game has
treewidth~$k$, then Odd has a strategy for selecting the sets~$S$
such that:
\begin{iteMize}{$\bullet$}
\item The simulation game always terminates after $k \cdot \log |V|$ many
rounds.
\item We never have to remember more than $3$ previous strategy profiles at
the same time.
\end{iteMize}
\noindent These two properties are sufficient to show that our alternating
Turing machine meets the required time and space bounds.

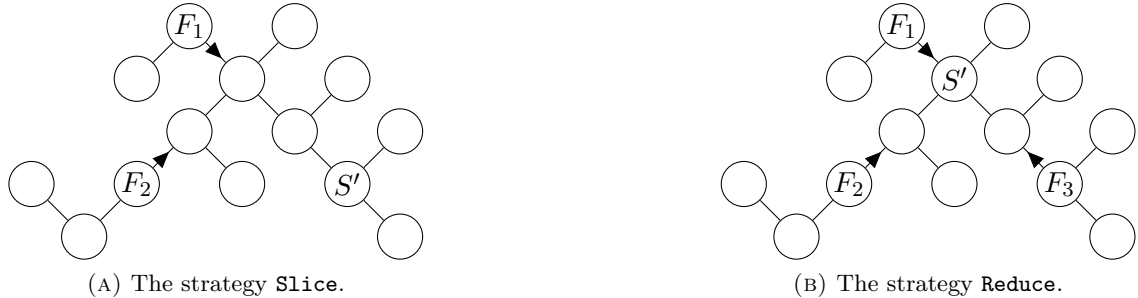
\begin{figure}
\subfloat[The strategy $\slice$.]
{
	\label{fig:slice}
	\begin{tikzpicture}[>=stealth,scale=0.7]
	\node (1) at (0, 0) [separator] {};
	\node (2) at (-1,1) [separator] {$F_1$}
		edge[->-] (1);
	\node (3) at (-2,0) [separator] {}
		edge (2);
	\node (4) at (-1,-1) [separator] {}
		edge (1);
	\node (5) at (0,-2) [separator] {}
		edge (4);
	\node (6) at (-2,-2) [separator] {$F_2$}
		edge[->-] (4);
	\node (7) at (-3,-3) [separator] {}
		edge (6);
	\node (8) at (-4,-2) [separator] {}
		edge (7);
	\node (9) at (1,1) [separator] {}
		edge (1);
	\node (10) at (1,-1) [separator] {}
		edge (1);
	\node (11) at (2,0) [separator] {}
		edge (10);
	\node (12) at (2,-2) [separator] {$S'$}
		edge (10);
	\node (13) at (3,-1) [separator] {}
		edge (12);
	\node (15) at (3,-3) [separator] {}
		edge (12);
	\end{tikzpicture}
}
\hfill
\subfloat[The strategy $\reduce$.]
{
	\label{fig:reduce}
	\begin{tikzpicture}[>=stealth,scale=0.7]
	\node (1) at (0, 0) [separator] {$S'$};
	\node (2) at (-1,1) [separator] {$F_1$}
		edge[->-] (1);
	\node (3) at (-2,0) [separator] {}
		edge (2);
	\node (4) at (-1,-1) [separator] {}
		edge (1);
	\node (5) at (0,-2) [separator] {}
		edge (4);
	\node (6) at (-2,-2) [separator] {$F_2$}
		edge[->-] (4);
	\node (7) at (-3,-3) [separator] {}
		edge (6);
	\node (8) at (-4,-2) [separator] {}
		edge (7);
	\node (9) at (1,1) [separator] {}
		edge (1);
	\node (10) at (1,-1) [separator] {}
		edge (1);
	\node (11) at (2,0) [separator] {}
		edge (10);
	\node (12) at (2,-2) [separator] {$F_3$}
		edge[->-] (10);
	\node (13) at (3,-1) [separator] {}
		edge (12);
	\node (15) at (3,-3) [separator] {}
		edge (12);
	\end{tikzpicture}
}
\caption{Odd's strategies for choosing $S$ in the simulation game.}
\label{fig:spaceexample}
\end{figure}

We now outline the strategy for Odd that achieves these two properties. In fact,
this strategy consists of two different strategies. The first strategy is called
$\slice$, and is shown in Figure~\textsc{\ref{fig:slice}}. The figure shows the nodes in
the tree decomposition, which means that each circle represents a set of
vertices in the parity game. The nodes $F_1$ and $F_2$ represent two previous
strategy profiles that we have remembered. The strategy $\slice$ is required to
select one of the nodes between $F_1$ and $F_2$. We do so
using the following well known lemma about trees.
\begin{lem}
\label{lem:splittree}
For every tree $T = (I, J)$ with $|I| \ge 3$, there is an $i \in I$ such
that removing~$i$ from~$T$ splits~$T$ into parts, where each part has at most
$\frac{2}{3}|I|$ vertices.
\end{lem}
We define the strategy $\slice$ to select a separator that satisfies
Lemma~\ref{lem:splittree}. For example, in Figure~\textsc{\ref{fig:slice}} the
separator $S'$ splits the tree into~$3$ parts, having $6$, $1$, and $1$
vertices, respectively. Since there were originally $9$ nodes
between $F_1$ and $F_2$, we have that $S'$ satisfies Lemma~\ref{lem:splittree}. 

The second strategy is called $\reduce$, and it is shown in
Figure~\textsc{\ref{fig:reduce}}. It is used whenever we have remembered three strategy
profiles. It selects the unique vertex that lies on the paths between them. It
can be seen in Figure~\textsc{\ref{fig:reduce}} that the set $S'$ lies on the unique
vertex that connects $F_1$, $F_2$, and $F_3$. The purpose of this strategy is to
reduce the number of strategy profiles that we must remember. It can be seen
that, no matter how the game on $S'$ ends, we will be able to forget at least
two of the three strategy profiles, while adding only one new strategy profile
for $S'$.

Odd's overall strategy combines these two sub-strategies: we use $\slice$ until
three strategy profiles have been remembered, and then we switch between
$\reduce$ and $\slice$. Applying this strategy ensures that we use $\slice$ at
least half of the time, and so the game must end after at most $O(\log |V|)$
simulation games have been played. The use of $\reduce$ ensures that we never
have to remember more than three strategy profiles.

Once these properties have been established, it is then fairly straightforward
to show that the game can be implemented by an alternating Turing machine in
$O(k^2 \cdot (\log |V|)^2)$ time and $O(k \cdot \log |V|)$ space, which then
immediately gives containment in NC$^2$.

\section{Simulated Parity Games}
\label{sec:game}

In this section we describe the simulation game, which will be used in
all of our subsequent results. We begin by formally defining the simulation
game. Since each of our results requires a slightly different version of the
simulation game, our definitions will be parameterized by two functions:
$\nextfun$ and $\history$, and each of our results will provide their own
versions of these functions. In the second part of this section, we will show
that, no matter which functions are chosen, if the simulation game terminates,
then it determines the correct winner. Then, for each of our results, we 
provide a proof of termination for the relevant $\nextfun$ and $\history$ functions.

We begin by defining \emph{records}, which allow us to remember the outcome of
previous games. A record is a triple $(F, p, P)$, where $F \subseteq V$ is a set
of vertices,~$P$ is a strategy profile for~$F$, and~$p \in D$ is the largest
priority that has been seen since~$P$ was rejected. When a record is created, we
will use $(F, -, P)$ to indicate that no priority has been seen since $P$ was
rejected. Given a record $(F, p, P)$, and a priority $p'$, we define
$\update((F, p, P), p') = (F, \max(p, p'), P)$, where we have $\max(-, p') =
p'$. A \emph{history} is an ordered sequence of records. Given a history
$\mathcal{F}$ and a priority $p'$, we define: $\update(\mathcal{F}, p') =
\{\update((F, p, P), p') \; : \; (F, p, P) \in \mathcal{F}\}$.

We can now define the format for the two functions $\nextfun$ and $\history$.
The function $\nextfun(S, v, \mathcal{F})$ takes a set of vertices $S \subseteq
V$, a vertex $v \in V$, and a history $\mathcal{F}$. It is required to return a
set of vertices $S' \subseteq V$. The history function $\history(\mathcal{F})$
allows us to delete certain records from the history $\mathcal{F}$. More
formally, the function is required to return a history $\mathcal{F'}$ with
$\mathcal{F'} \subseteq \mathcal{F}$.

We can now formally define the simulation game. Let $G$ be a parity game. Given
$S \subseteq V$,  $s \in V$, a history $\mathcal{F}$, and two functions
$\nextfun$ and $\history$, we define the game $\simulate_G(S, \mathcal{F}, s,
\nextfun, \history)$ as follows. The game maintains a variable~$c$ to store the
current vertex. It also maintains a sequence of triples $\Pi$, where each entry
is of the form $(v, p, u)$ with $v, u \in V$ and $p \in D$. The sequence $\Pi$
represents the simulated path that the two players form during the game.
 
To define our game, we will require notation to handle the simulated path $\Pi$.
We extend $\maxpri$ for paths of the form $\Pi = \langle (v_1, p_1, v_2), (v_2,
p_2, v_3), \dots, (v_{j-1}, p_{j-1}, v_j) \rangle$ by defining $\maxpri(\Pi) =
\max\{p_i \; : \; 1 \le i \le j-1\}$. Furthermore, if a path~$\Pi$ consists of
an initial path $\langle (v_1, p_1, v_2), \dots, (v_{j-1}, p_{j-1}, v_j)\rangle$
followed by a cycle $\langle (v_{j}, p_j, v_{j+1}), \dots, (v_{c-1}, p_{c-1},
v_j)$, then we define $\winner(\Pi)$ to be Even if $\max(\{p_i \; : \; j \le i
\le c - 1\})$ is even, and Odd otherwise.

We are now able to define $\simulate_G(S, \mathcal{F}, s, \nextfun, \history)$.
The game is played in rounds. The first round of the game is slightly different,
because it requires a special initialization procedure that will be introduced
later. Every other round proceeds as follows.

\begin{enumerate}[(1)]
\item \label{lbl:game1} An edge $(c, v)$ is selected by Even if $c \in V_0$ or
by Odd if $c \in V_1$. 

\item \label{lbl:game2} If $v \in S$ or if $v \in F$ for some $(F, p, P) \in
\mathcal{F}$, then the tuple $(c, \pri(v), v)$ is added to $\Pi$, the vertex~$c$
is set to~$v$, and the game moves to Step~\ref{lbl:game5}.

\item \label{lbl:game3} Even gives a strategy profile~$P'$ for the vertex~$v$
and the set~$S$.

\item \label{lbl:game4} Odd can either play \acc\ for some vertex $u \in S$
with $P'(u) \ne -$, or play \rej. 
\begin{iteMize}{$\bullet$}
\item If Odd plays \acc, then $(c, \max(\pri(v), P'(u)), u)$ is appended to
$\Pi$, and~$c$ is set to~$u$. 

\item If Odd plays \rej, then:
\begin{iteMize}{$-$}
\item The history $\mathcal{F'}$ is obtained by computing $\update(\mathcal{F},
\maxpri(\Pi))$.
\item The history $\mathcal{F''}$ is obtained by adding  $(S, -, P')$ to
the end of $\mathcal{F'}$.
\item The history $\mathcal{F'''}$ is obtained by computing
$\history(\mathcal{F''})$.

\item The set $S'$ is obtained by computing $\nextfun(S, v, \mathcal{F'''})$.
\item The winner of the game is the winner of $\simulate_G(S', \mathcal{F'''},
v, \nextfun, \history)$.
\end{iteMize}
\end{iteMize}

\item \label{lbl:game5} If $c \in F$ for some $(F, p, P) \in \mathcal{F}$, then
the game stops. Let $(F, p, P)$ be the final record in $\mathcal{F}$ such that
$c \in F$. If $P(c) = -$ then Odd wins the game. Otherwise, let $p' =
\max(\maxpri(\Pi), p)$. Even wins the game if~$p' \succeq P(c)$ and Odd wins
if~$p' \prec P(c)$.

\item \label{lbl:game6} If $\Pi$ ends in a cycle, then the winner of the
game is $\winner(\Pi)$.
\end{enumerate}

\noindent Note that, whenever we move to a new simulation game in
Step~\ref{lbl:game4}, the variables $c$ and $\Pi$ are reset to their initial
values. We will omit the subscripted parity game~$G$ from $\simulate_G(S,
\mathcal{F}, s, \nextfun, \history)$ when it is clear from the context.

As we have mentioned, the first round is slightly different. This is because we
allow the starting vertex~$s$ to be any vertex in $V$. Thus, we need a procedure
to initialize the variable $c$. If we happen to have $s \in S$, then we can set
$c = s$, and start the game in Step~\ref{lbl:game1} as normal. If $s \notin S$,
then we start the game in Step~\ref{lbl:game3} with $v = s$. In other words, the
game begins by allowing Even to give a strategy profile for $v$ and $S$. It can
be seen that, if this profile is accepted by Odd, then $c$ is set to some vertex
$u \in S$. Thus, after this initialization procedure, the game can continue on
as normal.

\subsection{Strategies for the simulation game}

Let $\sigma \in \Sigma_0$ be an Even strategy for the original parity game. We
define a strategy for Even in $\simulate$ called $\follow_0(\sigma)$, which
follows the moves made by $\sigma$. More formally, $\follow_0(\sigma)$ does the
following:
\begin{iteMize}{$\bullet$}
\item If Even is required to select an edge in Step~\ref{lbl:game1}, then the
edge $(c, \sigma(c))$ is selected.
\item In Step~\ref{lbl:game3}, the strategy selects $P' = \profile(\sigma, v, S)$.
\end{iteMize}

\noindent On the other hand, let $\tau \in \Sigma_1$ be a strategy for Odd in the original
parity game. We define a strategy for Odd in $\simulate$ called
$\follow_1(\tau)$, which is analogous to the strategy $\follow_0(\sigma)$ that
we have just defined. Odd's strategy is more complex, because it must decide
whether the strategy profile proposed by Even in Step~\ref{lbl:game3} should be
accepted. To aid in this, we use the following definition. Let $P_\tau =
\profile(\tau, s, S)$ for some $s \in V$ and $S \subseteq V$. If~$P$ is a
strategy profile for $S$ given by Even in Step~\ref{lbl:game3}, then $\tau$
\emph{refutes}~$P$ if one of the following conditions holds.

\begin{iteMize}{$\bullet$}
\item For all $u \in S$, we have $P(u) = -$.
\item For all $u \in S$, we have either $P_\tau(u) \ne -$ and $P(u) = -$, or
$P_\tau(u) \prec P(u)$.
\end{iteMize}

\noindent 
The first condition is necessary to deal with cases where Even gives a strategy
profile $P$ with $P(u) = -$ for all $u \in S$. In this case, Odd can not play
\acc\, and therefore has no choice but to reject $P$ in Step~\ref{lbl:game4}.
The second condition detects whether Even gives a false strategy profile: the
condition ensures that if Odd rejects $P$, and if in a subsequent simulation
game we arrive back at a vertex $u \in S$, then Odd will win in
Step~\ref{lbl:game5}.

The strategy $\follow_1(\tau)$ follows~$\tau$ for the vertices in~$S$, and
whenever Even gives a strategy profile~$P$, Odd plays \rej\ only when~$\tau$
refutes $P$. If $\tau$ does not refute $P$, then there must be at least one $u
\in S$ such that $P_\tau(u) \succeq P(u)$. Odd selects one such~$u$ and plays
\acc\ for it. Formally, we define $\follow_1(\tau)$ as follows:

\begin{iteMize}{$\bullet$}
\item If Odd selects an edge in Step~\ref{lbl:game1}, then Odd selects $(c,
\tau(c))$.
\item In Step~\ref{lbl:game4}, Odd's decision is based on the strategy profile
$P'$. If $\tau$ refutes $P'$, then odd plays \rej. Otherwise, Odd selects
a vertex $u \in F$ such that $\profile(\tau, v, S)(u) \succeq
P'(u)$ and plays \acc\ for $u$.
\end{iteMize}

\subsection{Plays in the simulation game}
We now describe the outcome when the simulation game
$\simulate(S, \mathcal{F}, s, \nextfun, \history)$ is played. The
players begin by playing $G_1 = \simulate(S, \mathcal{F}, s, \nextfun,
\history)$. This game either ends in
Step~\ref{lbl:game5}, Step~\ref{lbl:game6}, or when Odd plays \rej\ in
Step~\ref{lbl:game4} and causes the players to move to a new game $G_2$. The
outcome can be described as a (potentially infinite) sequence of games:
\begin{equation*}
\langle \mathcal{S}_1 = \simulate(S_1, \mathcal{F}_1, s_1, \nextfun, \history),
\mathcal{S}_2 =
\simulate(S_2, \mathcal{F}_2, s_2, \nextfun, \history), \dots \rangle.
\end{equation*}
We will call this a \emph{play} of the simulation game.

We define some notation for working with simulation game plays. Each of the
games $\mathcal{S}_i$ maintained a variable $\Pi$. At the start of the game
$\mathcal{S}_i$, the variable $\Pi$ was empty, and at the end of the game, $\Pi$
contained a sequence of tuples, which represented a sequence of real and
simulated moves in $S_i$. For each $i$, we define~$\Pi_i$ to be the state of the
variable~$\Pi$ at the end of the game~$\mathcal{S}_i$. Furthermore, for each
$i$, we will represent $\Pi_i$ as:
\begin{equation*}
\langle (v_{i, 1}, p_{i, 1}, u_{i, 1}), (v_{i, 2}, p_{i, 2}, u_{i, 2}), \dots,
(v_{i, |\Pi_i|}, p_{i, |\Pi_i|}, u_{i, |\Pi_i|}) \rangle.
\end{equation*}

We are interested in plays of the simulation game that arise when either Even
plays $\follow_0(\sigma)$ for some strategy $\sigma \in \Sigma_0$, or Odd plays
$\follow_1(\tau)$ for some strategy $\tau \in \Sigma_1$. For the sake of
exposition, we focus on player Even, but we will prove our results for both
players. We want to show that, if Even uses $\follow_0(\sigma)$, then we can translate the paths
$\Pi_i$ in the simulation game to paths in the original parity game. More
precisely, for each tuple $(v_{i, j}, p_{i,j}, u_{i,j})$ in $\Pi_i$, we want to
construct a path $\pi_{i, j}$ in $G \restriction \sigma$. For the portions of
$\Pi_i$ that correspond to real edges in the parity game, this is easy.
Formally, if $(v_{i, j}, p_{i,j}, u_{i,j})$ was not added to $\Pi_i$ by a
simulated move, then we can define $\pi_{i, j} = \langle v_{i, j}, u_{i, j}
\rangle$. This path is obviously a path in $G \restriction \sigma$. The
following lemma shows that paths $\pi_{i, j}$ can be defined for the simulated
moves in $\Pi_i$.



\begin{lem}
\label{lem:pitopi}
We have two properties:
\begin{iteMize}{$\bullet$}
\item If there is a strategy $\sigma \in \Sigma_0$ such that:
\begin{equation*}
\profile(\sigma, v_{i, j}, S_i)(u_{i,j}) = p_{i, j} \ne -,
\end{equation*}
for some~$i$ and~$j$, then there exists a path $\pi_{i, j} \in \paths(\sigma,
v_{i,j}, S_i, u_{i,j})$ such that 
\[\maxpri(\pi_{i, j}) = p_{i, j}.\]
\item
If there is a strategy $\tau \in \Sigma_1$ such that: \begin{equation*}
\profile(\tau, v_{i, j}, S_i)(u_{i,j}) \succeq p_{i, j} \ne -, \end{equation*}
for some~$i$ and~$j$, then there exists a path $\pi_{i, j} \in \paths(\tau,
v_{i,j}, S_i, u_{i,j})$ such that 
\[\maxpri(\pi_{i, j}) \succeq p_{i, j}.\]
\end{iteMize}
\end{lem}
\begin{proof}
The first part of this lemma follows directly from the definition of a strategy
profile. Since we know that $\profile(\sigma, v_{i, j}, S_i)(u_{i,j}) \ne -$, we
have that $\paths(\sigma, v_{i, j}, S_i, u_{i, j}) \ne \emptyset$. Furthermore,
since $\profile(\sigma, v_{i, j}, S_i)(u_{i,j})$ considers the minimum path
according to the $\preceq$-ordering, there must exist a path $\pi_{i, j}$ in
$\paths(\sigma, v_{i, j}, S_i, u_{i,j})$ such that $\maxpri(\pi_{i,j}) = p_{i,
j}$. 

The second part is very similar to the first. Since we know that $\profile(\tau,
v_{i, j}, S_i)(u_{i,j}) \ne -$, we have that $\paths(\tau, v_{i, j}, S_i, u_{i,
j}) \ne \emptyset$. Furthermore, since $\profile(\tau, v_{i, j}, S_i)(u_{i,j})$
considers the maximum path according to the $\preceq$-ordering, there must exist
a path $\pi_{i, j}$ in $\paths(\tau, v_{i, j}, S_i, u_{i,j})$ such that
$\maxpri(\pi_{i,j}) = \profile(\tau, v_{i, j}, S_i)(u_{i,j}) \succeq p_{i, j}$. 
\end{proof}

For the rest of this section, we will assume that either Even always plays
$\follow_0(\sigma)$ for some strategy $\sigma \in \Sigma_0$, or that Odd always
plays $\follow_1(\tau)$ for some strategy $\tau \in \Sigma_1$. We will prove
properties about the paths~$p_{i, j}$ that are obtained from
Lemma~\ref{lem:pitopi}.

Firstly, we note that if $\pi_{i, j}$ is concatenated with $\pi_{i, j+1}$, then
the result is still a path in either $G \restriction \sigma$ or $G \restriction
\tau$. This is because Steps~\ref{lbl:game2} and~\ref{lbl:game4} of the
simulation game ensure that the final vertex in $\pi_{i, j}$ is the first vertex
in $\pi_{i, j+1}$. Similarly, the concatenation of $\pi_{i, |\pi_i|}$ and
$\pi_{i+1, 1}$ always produces a path in either $G \restriction \sigma$, or $G
\restriction \tau$,  because Step~\ref{lbl:game4} of the simulation game ensures
that the final vertex in $\pi_{i, |\Pi_i|}$ must have an edge to the first
vertex in $\pi_{i+1, 1}$.

With these properties in mind, we define $\pi_{(i, j) \rightarrow (i', j')}$,
for two pairs of integers where, either $i' > i$, or both $i = i'$ and $j' > j$,
to be the concatenation of $\pi_{i, j}$ through $\pi_{i', j'}$. More formally,
let $\circ$ be the operator that concatenates two paths.
We define:
\begin{equation*}
\pi_i = \pi_{i, 1} \circ \pi_{i, 2} \circ \dots \circ \pi_{i, |\pi_i|},
\end{equation*}
to be the path from the first vertex in $\pi_{i, 1}$ to the final vertex in
$\pi_{i, |\pi_i|}$. We then define
\begin{equation*}
\pi_{(i, j) \rightarrow (i', j')} = \pi_{i, j} \circ \dots \circ \pi_{i, |\pi_i|}
\circ \pi_{i+1} \circ \dots \circ \pi_{i'-1} \circ \pi_{i', 1} \circ \dots \circ
\pi_{i', j'},
\end{equation*}
to be the concatenation of the paths between $\pi_{i, j}$ and $\pi_{i',j'}$.

The following lemma regarding $\maxpri(\pi_{(i, j) \rightarrow (i',
j')})$ follows easily from Lemma~\ref{lem:pitopi}.
\begin{lem}
\label{lem:concatpri}
If $\sigma \in \Sigma_0$ is an Even strategy, and $\langle G_1 \dots G_k
\rangle$ is obtained when Even plays $\follow_0(\sigma)$, then we have:
\begin{multline*}
\maxpri(\pi_{(i, j) \rightarrow (i', j')}) = \max\{p_{x, y} \; : \;
i \le x \le i' \text{ and } \\ (x = i) \Rightarrow j \le y \text{ and } (x = i')
\Rightarrow y \le j'\}.
\end{multline*}
On the other hand, if $\tau \in \Sigma_1$ is an Odd strategy, and $\langle G_1
\dots G_k \rangle$ is obtained when Odd plays $\follow_1(\tau)$, then we have:
\begin{multline*}
\maxpri(\pi_{(i, j) \rightarrow (i', j')}) \succeq \max\{p_{x, y} \; : \;
i \le x \le i' \text{ and } \\ (x = i) \Rightarrow j \le y \text{ and } (x = i')
\Rightarrow y \le j'\}.
\end{multline*}
\end{lem}

\subsection{Winning strategies and the simulation game} 
\label{sec:wingame}

In this section, we show that if one of the two players follows a winning
strategy for the original parity game, then that player cannot lose the
simulation game. More formally, we show that, if $\sigma \in \Sigma_0$ is a
winning strategy for Even for some vertex $s \in V$, and if Even plays
$\follow_0(\sigma)$,  then Even cannot lose $\simulate(S, \mathcal{F}, s,
\nextfun, \history)$, for all $S$, $\mathcal{F}$, $\nextfun$, and $\history$. We
also show the analogous result for Odd. 



We begin by showing that, if Even never lies about his strategy profile in
Step~\ref{lbl:game3}, then Even can never lose the game in Step~\ref{lbl:game5}.

\begin{lem}
\label{lem:evenwinspace}
If Even plays $\follow_0(\sigma)$ for some $\sigma \in \Sigma_0$, then Even
can never lose the simulation game in Step~\ref{lbl:game5}.
\end{lem}
\begin{proof}
Assume, for the sake of contradiction, that Even loses the simulation game in
Step~\ref{lbl:game5} while playing $\follow_0(\sigma)$. Let
\begin{equation*}
\langle \mathcal{S}_1 = \simulate(S_1, \mathcal{F}_1, s_1, \nextfun, \history),
\dots, 
\mathcal{S}_k =
\simulate(S_k, \mathcal{F}_k, s_k, \nextfun, \history) \rangle
\end{equation*}
be the play of the simulation game, where Even loses $\mathcal{S}_k$ in
Step~\ref{lbl:game5}.

Let $(F, p, P)$ be the record used to decide the winner in Step~\ref{lbl:game5}.
The fact that Odd wins in Step~\ref{lbl:game5} implies that $\Pi_k$ ends at some
vertex $u \in F$, and that either $P(u) = -$ or $\max( \maxpri(\Pi), p) \prec
P(u)$. Let $i$ be the index of the game $\mathcal{S}_i$ in which $(F, -, P)$ is
contained in $\mathcal{F}_{i+1} \setminus \mathcal{F}_{i}$. 

We begin with the case where $P(u) = -$. In this case, the fact that
$\follow_0(\sigma)$ only ever gives $\profile(\sigma, s, S)$ for $P'$ in
Step~\ref{lbl:game3} allows us to invoke
Lemma~\ref{lem:pitopi} to conclude that the path $\pi_{(i + 1, 1) \rightarrow
(k, |\Pi_k|)}$ is a path in $G \restriction \sigma$. Note that the path starts
at~$s_{i+1}$, and that~$u$ is the first vertex in $F$ that is visited by
$\pi_{(i+1, 1) \rightarrow (k, |\Pi_k|)}$. Hence we must have $\pi_{(i+1, 1)
\rightarrow (k, |\Pi_k|)} \in \paths(\sigma, s_{i+1}, F, u)$. Therefore, the
fact that the set $\paths(\sigma, s_{i+1}, F, u)$ is non-empty contradicts the
fact that $P(u) = \profile(\sigma, s_{i+1}, F)(u) = -$.

We now consider the case where Odd wins in Step~\ref{lbl:game5} because $\Pi_k$
ends at some vertex $u \in F$ and $\max(\maxpri(\Pi), p) \prec P(u)$. In this
case we can again invoke Lemma~\ref{lem:pitopi} to argue that there is a path
$\pi_{(i+1, 1) \rightarrow (k, |\Pi_k|)}$ in $G \restriction \sigma$, and we can
use the same arguments as in the previous case to argue that this path is
contained in $\paths(\sigma, s_{i+1}, F_i, u)$. Moreover,
Lemma~\ref{lem:concatpri} implies that 
\begin{equation*}
p' = \max(p, \maxpri(\Pi)) = \maxpri(\pi_{(i+1, 1) \rightarrow (k, |\pi_{k}|)}).
\end{equation*}
Since $p' \prec P(u)$, and $P(u) = \profile(\sigma, s_{i+1}, F)$, we have the
following inequality:
\begin{equation*}
\maxpri(\pi_{(i+1, 1) \rightarrow (k, |\Pi_k|)}) \prec \min_{\preceq}
\{\maxpri(\pi) \: : \: \pi \in \paths(\sigma, s_{i+1}, F, u)\}.
\end{equation*}
Since~$\pi_{(i+1, 1) \rightarrow (k, |\Pi_k|)}$ is contained in $\paths(\sigma,
s_{i+1}, F, u)$ this inequality is impossible, which yields the required
contradiction.
\end{proof}

We now show a corresponding property for Odd. 

\begin{lem}
\label{lem:oddwinspace}
If Odd plays $\follow_1(\tau)$ for some strategy $\tau \in \Sigma_1$, then Odd
can never lose the simulation game in Step~\ref{lbl:game5}.
\end{lem}
\begin{proof}
Assume, for the sake of contradiction, that Odd loses the simulation game in
Step~\ref{lbl:game5} while playing $\follow_1(\tau)$. Let
\begin{equation*}
\langle \mathcal{S}_1 = \simulate(S_1, \mathcal{F}_1, s_1, \nextfun, \history),
\dots, 
\mathcal{S}_k =
\simulate(S_k, \mathcal{F}_k, s_k, \nextfun, \history) \rangle
\end{equation*}
be the play of the simulation game, where Odd loses $\mathcal{S}_k$ in
Step~\ref{lbl:game5}.

Let $(F, p, P)$ be the record that is used to determine the winner in
Step~\ref{lbl:game5}. The fact that Even wins in Step~\ref{lbl:game5} implies
that $\Pi_k$ ends at some vertex $u \in F$, and that $\max(\maxpri(\Pi), p)
\succeq P(u)$.  Let $i$ be the index of the game $\mathcal{S}_i$ in which $(F,
-, P)$ is contained in $\mathcal{F}_{i+1} \setminus \mathcal{F}_{i}$.
Since~$\follow_1(\tau)$ only ever plays~$\rej$ in the case where~$\tau$
refutes~$P'$, we can apply Lemma~\ref{lem:pitopi} to obtain a path $\pi_{(i + 1,
1) \rightarrow (k, |\Pi_k|)}$ in $G \restriction \tau$. Note that the path
starts at the vertex $s_{i+1}$, and that~$u$ is the first vertex in $F$ that is
visited by $\pi_{(i+1, 1) \rightarrow (k, |\Pi_k|)}$. Hence we must have
$\pi_{(i+1, 1) \rightarrow (k, |\Pi_k|)} \in \paths(\tau, s_{i+1}, F, u)$.
Furthermore, we can use Lemma~\ref{lem:concatpri} to conclude:
\begin{equation*}
\maxpri(\pi_{(i+1, 1) \rightarrow (k, |\Pi_k|)}) \succeq \max(\maxpri(\Pi), p)
\succeq P(u).
\end{equation*}
Since~$\tau$ refutes~$P$, we must have $P(u) \succ \profile(\tau, s_{i+1},
F)(u)$. Combining all of these facts yields:
\begin{align*}
\maxpri(\pi_{(i+1, 1) \rightarrow (k, |\Pi_k|)}) &\succeq P(u) \\
&\succ \profile(\tau, s_{i+1}, F)(u) \\
&= \max_{\preceq} \{\maxpri(\pi) \; : \; \pi \in \paths(\tau, s_{i+1}, F, u)\}.
\end{align*}
However, since $\pi_{(i+1, 1) \rightarrow (k, |\Pi_k|)} \in \paths(\tau,
s_{i+1}, F_i, u)$, this inequality is impossible. 
\end{proof}

We now turn our attention to Step~\ref{lbl:game6} of the simulation game. The
following pair of lemmas will be used to show that, if one of the two players
follows a winning strategy, then that player cannot lose the simulation game in
Step~\ref{lbl:game6}.

\begin{lem}
\label{lem:cycle2}
Let $\sigma \in \Sigma_0$ be an Even winning strategy for $s \in V$. If Even
plays $\follow_0(\sigma)$, then Even cannot lose $\simulate(S, \mathcal{F}, s,
\nextfun, \history)$ is Step~\ref{lbl:game6}, for all choices of $S$,
$\mathcal{F}$, $\nextfun$, and $\history$.
\end{lem}
\begin{proof}
Assume, for the sake of contradiction, that Even loses the simulation game in
Step~\ref{lbl:game6} while playing $\follow_0(\sigma)$. Let
\begin{equation*}
\langle \mathcal{S}_1 = \simulate(S_1, \mathcal{F}_1, s_1, \nextfun, \history),
\dots, 
\mathcal{S}_k =
\simulate(S_k, \mathcal{F}_k, s_k, \nextfun, \history) \rangle
\end{equation*}
be the play of the simulation game, where Even loses $\mathcal{S}_k$ in
Step~\ref{lbl:game6}.

Since Even loses in Step~\ref{lbl:game6}, there must be some index~$i$ such
that~$\langle \Pi_{k, i}, \dots, \Pi_{k, |\Pi_k|} \rangle$ forms a cycle
in~$S_k$, and that the highest priority on this cycle is odd. Since Even plays
$\follow_0(\sigma)$, we can use Lemma~\ref{lem:pitopi} to conclude that
$\pi_{(k, i) \rightarrow (k, |\Pi_k|)}$ is a cycle in $G \restriction \sigma$,
and we can apply Lemma~\ref{lem:concatpri} to conclude that $\maxpri(\pi_{(k, i)
\rightarrow (k, |\Pi_k|)})$ is odd. Furthermore, we have that $\pi_{(1,1)
\rightarrow k, (i-1)}$ is a path from $s$ to the cycle $\pi_{(k, i) \rightarrow
(k, |\Pi_k|)}$ in $G \restriction \sigma$. This implies that~$\sigma$ is not a
winning strategy for~$s$, which provides the required contradiction. 
\end{proof}

\begin{lem}
\label{lem:cycle3}
Let $\tau \in \Sigma_1$ be an Odd winning strategy for $s \in V$. If Odd
plays $\follow_1(\tau)$, then Odd cannot lose $\simulate(S, \mathcal{F}, s,
\nextfun, \history)$ is Step~\ref{lbl:game6}, for all choices of $S$,
$\mathcal{F}$, $\nextfun$, and $\history$.
\end{lem}
\begin{proof}
Assume, for the sake of contradiction, that Odd loses the simulation game in
Step~\ref{lbl:game6} while playing $\follow_1(\tau)$. Let
\begin{equation*}
\langle \mathcal{S}_1 = \simulate(S_1, \mathcal{F}_1, s_1, \nextfun, \history),
\dots, 
\mathcal{S}_k =
\simulate(S_k, \mathcal{F}_k, s_k, \nextfun, \history) \rangle
\end{equation*}
be the play of the simulation game, where Odd loses $\mathcal{S}_k$ in
Step~\ref{lbl:game6}.

Since Odd loses in Step~\ref{lbl:game6}, there must be some index~$i$ such
that~$\langle \Pi_{k, i}, \dots, \Pi_{k, |\Pi_k|} \rangle$ forms a cycle
in~$S_k$, and that that the highest priority on this cycle is even. Since Odd
plays $\follow_1(\tau)$, we can use Lemma~\ref{lem:pitopi} to conclude that
$\pi_{(k, i) \rightarrow (k, |\Pi_k|)}$ is a cycle in $G \restriction \tau$, and
we can apply Lemma~\ref{lem:concatpri} to conclude that $\maxpri(\pi_{(k, i)
\rightarrow (k, |\Pi_k|)}) \succeq \maxpri(\Pi_k)$. Since  $\maxpri(\Pi_k)$ is
even, every priority $p$ with $p \succeq \maxpri(\Pi_k)$ must also be even.
Hence, we can conclude that $\maxpri(\pi_{(k, i) \rightarrow (k, |\Pi_k|)})$ is
even. Furthermore, we have that $\pi_{(1,1) \rightarrow (k, i-1)}$ is a path
from $s_1$ to the cycle $\pi_{(k, i) \rightarrow k, (|\Pi_k|)}$ in $G
\restriction \tau$. This implies that~$\tau$ is not a winning strategy for~$s$,
which provides the required contradiction. 
\end{proof}

When combined, Lemmas~\ref{lem:evenwinspace} through~\ref{lem:cycle3} imply the
following property, which is the main result of this section.

\begin{lem}
\label{lem:correct}
If $\simulate_G(S, \mathcal{F}, s, \nextfun, \history)$ terminates, then it
correctly determines the winner of $s$.
\end{lem}

\section{Time Complexity Results for Parity Games}
\label{sec:time}

Let~$G$ be a parity game with a DAG decomposition $(\mathcal{D} = (I, J), X)$ of
width~$k$. In this section we give a version of the simulation game  that can be
solved in $(k + 3) \cdot \log |V| + (k + 2) \cdot \log(k) + (3k + 2) \cdot
\log(|D| + 1)$ space on an alternating Turing machine, and hence 
$O(|V|^{k+3} \cdot k^{k + 2} \cdot (|D| + 1)^{3k + 2})$ 
time on a deterministic Turing
machine.

\subsection{The functions \texorpdfstring{$\nextfun$}{Next} and \texorpdfstring{$\history$}{Hist}}
 Recall that the simulation
game can be customized by specifying the $\nextfun$ and $\history$ functions. We
begin by defining the versions of these functions that will be used for this
result. We start by defining the $\nextfun$ function. Recall, from
Section~\ref{sec:outline1}, that we begin by playing a simulation game on a
source in the DAG decomposition. To define this formally, we must be more
specific: if we want to determine the winner of some vertex $v \in V$, then we
must pick a source $i$ from the DAG decomposition such that $v \in \guarded(X_i)
\cup X_i$, and play a simulation game on $X_i$. 

The definition of $\nextfun$ follows the intuition that was outlined in
Section~\ref{sec:outline1}. Suppose that we are required to compute $\nextfun(S,
v, \mathcal{F})$, for some set $S = X_i$ with $i \in I$, some $v \in V$, and some
history $\mathcal{F}$. We find the DAG decomposition node $j \in I$ such that
$(i, j) \in J$, and $v \in \guarded(X_j) \cup X_j$, and set $\nextfun(S, v,
\mathcal{F}) = X_j$. After defining $\history$, we will prove
Lemma~\ref{lem:vinxi}, which shows that such a~$j$ always exists.

To define the $\history$ function, we use the observation that we made in
Section~\ref{sec:outline1}: we only need to remember at most one record. More
precisely, we only need to remember the record corresponding to the last
simulation game that was played. Formally, we can define the $\history$ function
as follows.
\begin{iteMize}{$\bullet$}
\item The first time that $\history(\mathcal{F})$ is called, $\mathcal{F}$ will
contain only one record. We do not delete this record, so we set
$\history(\mathcal{F}) = \mathcal{F}$.
\item Every other time that $\history(\mathcal{F})$ is called, $\mathcal{F}$
will contain exactly two records. Due to the definition of $\nextfun$, we know
that the two records in $\mathcal{F}$ have the form $(X_i, v_1,
P_1)$ and $(X_j, v_2, P_2)$ for some $i, j \in I$, some
$v_1, v_2 \in V$, and some strategy profiles $P_1$ and $P_2$.
Furthermore, we know that $(i, j) \in J$ is a directed edge in the DAG
decomposition. We set $\history(\mathcal{F}) = \{(X_j, v_2, P_2)\}$.
\end{iteMize}

\noindent We now show the correctness of the $\nextfun$ function.

\begin{lem}
\label{lem:vinxi}
Let $S = X_i$ for some $i \in I$, and $s \in \guarded(X_i) \cup X_i$. Suppose
that $\simulate(S, \emptyset, s, \nextfun, \history)$ ends when Odd rejects a
strategy profile for some vertex $v \in V$. There exists a $j \in I$ with $(i,
j) \in J$ and $v \in \guarded(X_j) \cup X_j$.
\end{lem}
\begin{proof}
The definitions of $\nextfun$ and $\history$ imply that there are two cases to
consider:
\begin{iteMize}{$\bullet$}
\item The case where $X_i$ is a sink, and $\mathcal{F}$ is empty. 
\item The case where $X_i$ is not a sink, and $\mathcal{F} = \{(X_k, v', P')\}$,
with $(k, i) \in J$.
\end{iteMize}
In the first case, the property follows from the definition of a guarded set: no
matter which edges $(c, v)$ are chosen in Step~\ref{lbl:game1}, we can never
move to a vertex $v \notin \guarded(X_i)$. In the second case, we can either
choose an edge $(c, v)$ with $v \in \guarded(X_i)$, or an edge $(c, v)$ with $v
\in X_k$. However, Step~\ref{lbl:game2} ensures that, if $v \in X_k$, then the
game skips directly to Step~\ref{lbl:game5}. Thus, in both cases, if Odd rejects
a strategy profile for some vertex $v$, with $(c, v) \in E$, then we have $v \in
\guarded(X_i)$.

To prove our claim, we can apply the properties given in
Definition~\ref{def:dagwidth}. They imply that, if there is an edge $(c, v) \in
E$ with $c \in X_i$ and $v \in \guarded(X_i) \setminus X_i$, then there must
exist an edge $(i, j) \in J$ with $v \in \guarded(X_j) \cup X_j$. 
\end{proof}

\subsection{Termination and correctness}
We now show that this version of the simulation game terminates, and
that it can be used to correctly determine the winner of the parity game. 
Recall that, in Section~\ref{sec:wingame}, we showed that if the simulation game
terminates, then it correctly determines the winner of the parity game. Thus,
we must show that the simulation game terminates when it is equipped with
$\nextfun$ and $\history$.

\begin{lem}
\label{lem:terminates}
If $S = X_i$ for some source $i \in I$, and $s \in \guarded(X_i) \cup X_i$ is a
starting vertex, then
$\simulate(S, \emptyset, s, \nextfun, \history)$ terminates.
\end{lem}
\begin{proof}
We begin by showing that Odd cannot play \rej\ in Step~\ref{lbl:game4}
infinitely many times.  This follows from the definition of $\nextfun$. This
functions ensures two properties:
\begin{iteMize}{$\bullet$}
\item Each simulation game is played on a set of vertices $S = X_j$
for some $j \in I$.
\item If odd plays \rej\ in a game on $S = X_j$, then the next game
will be played on $S' = X_k$, with $(j, k) \in J$.
\end{iteMize} 
These two properties ensure that, every time that Odd plays \rej, we take one
step along a directed path in the DAG decomposition. Since the DAG decomposition
contains no cycles, we know that Odd can play \rej\ at most $|I|$ times before
we reach a sink in the DAG decomposition.

The $\history$ function ensures that, once we reach a sink in the
DAG decomposition, the game must end. Suppose that we have reached a simulation
game on $S = X_j$ where $j \in I$ is a sink. The history function ensures
that $\mathcal{F}$ contains a record $(F, p, P)$ with $F = X_k$,
where $(k, j) \in J$. Therefore, since $X_j$ guards $X_i$, every
edge $(v, w)$ with $v \in S$ and $w \notin S$ has $w \in X_j$. Thus,
Step~\ref{lbl:game2} of the simulation game ensures that Odd cannot play \rej.

We now show that if Odd only plays \acc\ in a simulation game on $S$, then the
game must terminate after $|S|+1$ steps. This is because each time Odd plays
\acc\ for some vertex $u \in S$, the vertex $u$ is appended to~$\Pi$. Therefore,
if the game continues for $|S|+1$ steps, then Odd will be forced to play \acc\
for some vertex $u$ that has been visited by~$\Pi$, and the game will terminate
in Step~\ref{lbl:game6}. 
\end{proof}

Then, we can apply Lemma~\ref{lem:terminates} along with
Lemma~\ref{lem:correct} to imply the correctness
of this version of the simulation game.

\begin{lem}
\label{lem:verifyparity}
Let $v \in V$ be a vertex, and let $X_i \in X$ be a set of vertices, where $i$
is a source, and $v \in \guarded(X_i) \cup X_i$. Even wins $\simulate(X_i,
\emptyset, v, \nextfun, \history)$ if and only if $v \in W_0$.
\end{lem}

\subsection{Implementation on an alternating Turing machine}

In this section, we show how this version of the simulation game can be
implemented on an alternating Turing machine. We use the existential states of
the machine to simulate the moves of Even and the universal states of the
machine to simulate the moves of Odd. In this way, we ensure that the machine
will have an accepting run if and only if Even has a strategy for $\simulate_G(S,
\emptyset, s, \nextfun, \history)$ that forces a win no matter how Odd plays.

\begin{thm}
\label{thm:time}
There is an algorithm that, when given a parity game~$G$, a DAG decomposition
$(\mathcal{D} = (I, J), X)$ of width~$k$ for $G$, and a vertex $s$, terminates
in $O(|V|^{k+3} \cdot k^{k + 2} \cdot (|D| + 1)^{3k + 2})$ time and outputs the
winner of $s$.
\end{thm}
\begin{proof}
We show that $\simulate_G(S, \emptyset, s, \nextfun, \history)$, where $S$ is a
source in the DAG decomposition, and $s \in \guarded(S) \cup S$, can be
implemented on an alternating Turing machine in logarithmic space. We split our
analysis into two cases: the data that must be stored throughout the game, and
the data that must be stored to implement each iteration of the game.

We begin by considering the data that must be stored throughout. The simulation
game maintains the variables $S$, $\mathcal{F}$, $\Pi$ and $c$. Note
that, from the definitions of $\history$ and $\nextfun$, we have that $S = X_i$
for some $i \in I$. Moreover, we know that $\mathcal{F}$ contains exactly one
record $(F, p, P)$, where $F = X_j$ for some $j$ with $(j, i) \in J$. 
These variables can be stored using the following amount of space.
\begin{iteMize}{$\bullet$}
\item 
We claim that $S$ and $F$ can be remembered using $(k + 2) \log |V|$ bits.
This is because we always have $S = X_i$ and $F = X_j$ for some 
edge $(j, i) \in J$, and since
since $|J| \le |V|^{k+2}$, we need at most $(k+2) \log |V|$ bits to store $(i,
j)$. Remembering $(i, j)$ is sufficient to identify both $S$ and $F$.
\item Note that~$c$ is always a vertex in $S$. Since we know that $S = X_i$, we can represent $c$ as a number between $1$ and $k$. Thus, $c$ can be stored in at most $\log k$ bits.
\item The priority $p$ can be stored in $\log D$ bits.
\item The strategy profile $P$ contains $|F|$ mappings from the vertices of $F$
to either a priority $d$, or $-$. Thus, we can represent a strategy profile as a
list of length $k$, where each element of the list is a number between $1$ and
$|D|+1$: the numbers $1$ through $|D|$ encode their respective priorities, and
the number $|D|+1$ encodes $-$. Thus, $P$ can be stored using $k \cdot \log (|D|+1)$
bits.
\item Recall that the sequence $\Pi$ is reset to $\emptyset$ at the start of
each simulation game. This implies that $\Pi$ contains at most $|S|+1$ tuples of
the form $(v, p, u)$. Due to the initialization procedure of the simulation game
we can have $v \notin S$ and $u \in S$ in the first element of $\Pi$. However,
for every subsequent element, we always have $v \in S$ and $u \in S$. Suppose
that the first element of $\Pi$ does have $v \notin S$. We claim that we do not
need to remember the vertex $v$. Note that the simulation game only cares about
the largest priority along $\Pi$, and whether $\Pi$ forms a cycle. Since $v
\notin S$, Step~\ref{lbl:game2} prevents $\Pi$ from forming a cycle that
includes $v$. Thus, we do not need to store $v$.

Note also that, if $(v, p, u)$ appears in $\Pi$, then, the next element in the
sequence must be of the form $(u, p', u')$. Thus, for every tuple in $\Pi$ other
than the first element in the sequence, we only need to store the priority $p'$
and the vertex $u'$. In total, therefore, we can represent $\Pi$ by storing at
most $|S| + 1$ priorities, and at most $|S| + 1$ vertices in $S$. These can be
stored in $(k+1) \cdot \log |D|  + (k + 1) \cdot \log k$ bits.
\end{iteMize}

\noindent Next, we account for the space used during the execution of the game.
\begin{iteMize}{$\bullet$}
\item The vertex $v$ that is selected in Step~\ref{lbl:game1} requires $\log
|V|$ bits to store.
\item As with the strategy profile $P$, the strategy profile $P'$ given by Even
in Step~\ref{lbl:game3} can be stored in $k \cdot \log (|D|+1)$ bits.
\item We claim that no additional space is needed if Odd rejects the strategy
profile $P'$ in step \ref{lbl:game4}. To see why, note that $\update$ will
always delete the current record in $\mathcal{F}$. Thus, we can implement
Step~\ref{lbl:game4} by first deleting the existing record in $\mathcal{F}$, and
then reusing the space for the new record.
\item No extra space is needed by Steps~\ref{lbl:game2},~\ref{lbl:game5},
and~\ref{lbl:game6}.
\end{iteMize}

\noindent Thus, in total, we require 
\begin{equation*}
(k + 3) \cdot \log |V| + (k + 2) \cdot \log(k) + (3k + 2) \cdot \log(|D| +
1)
\end{equation*}
bits of storage for to implement the game on an alternating Turing machine.
Using the standard determinization techniques for alternating Turing machines,
we can obtain a deterministic algorithm that runs in $O(|V|^{k+3} \cdot k^{k
+ 2} \cdot (|D| + 1)^{3k + 2})$ time.
\end{proof}


\section{A Faster Algorithm For Parity Games With Bounded Treewidth}
\label{sec:complexity}

In Section~\ref{sec:time} we gave an 
$O(|V|^{k+3} \cdot k^{k + 2} \cdot (|D| + 1)^{3k + 2})$
 time algorithm for parity games with bounded DAG width. In this
section, we show that, if the parity game has bounded treewidth, then we can
produce a version of the simulation game that can be solved in 
$O(|V| \cdot (k+1)^{k+5} \cdot (|D|+1)^{3k + 5})$ 
time.

\subsection{Outline.} 

The version of the simulation game used in this section is
very similar to the version used in Section~\ref{sec:time}, but this time we
will be using a tree decomposition. In fact, we will convert the tree
decomposition to a \emph{rooted} tree decomposition.
Suppose that we want to determine the winner of some vertex $s$
in the parity game. We find a tree decomposition node $i$ such that $s \in
X_i$, and we declare $i$ to be the root of $T$. Then, we turn $T$ into a
directed tree, by orienting all edges away from the root $i$. This step is
necessary to obtain our desired complexity.

Recall that the factor of $|V|^{k+3}$ in Theorem~\ref{thm:time} arose from the
fact that we needed to use $(k+3) \cdot \log |V|$ bits of storage in the
alternating Turing machine:
\begin{iteMize}{$\bullet$}
\item We needed $(k+2) \cdot \log |V|$ bits to store an edge from the DAG
decomposition.
\item We needed $\log |V|$ bits to store the vertex $v$.
\end{iteMize}
\noindent The first observation that we make is that, when we consider parity
games with bounded treewidth, we can easily reduce the amount of space that is
required. Whereas a DAG decomposition can have up to $|V|^{k+2}$ edges, a tree
decomposition can have at most $|V| - 1$ edges. Thus, we can reduce the amount of
space needed by the alternating Turing machine to 
$2 \cdot \log |V| + (k + 2) \cdot \log(k) + (3k + 2) \cdot
\log(|D| + 1)$, and thereby obtain a $O(|V|^{2} \cdot k^{k + 2} \cdot (|D| +
1)^{3k + 2})$ time algorithm for solving parity games with bounded treewidth.

We can, however, do better than this. For the rest of this section, our goal is
to show that the $2 \cdot \log |V|$ term in our space bounds can be reduced to
$\log |V|$. This will lower the $|V|^2$ term in the complexity of the
algorithm to $|V|$. To do this, we concentrate on the amount of space used by
the variable $v$. The reason why we require $\log |V|$ bits to store~$v$ is that
the edge chosen in Step~\ref{lbl:game1} can potentially set $v$ to be any vertex
in the parity game. Figure~\ref{fig:subtree} illustrates this. It shows a
subtree of a tree decomposition: specifically, it shows the tree decomposition
nodes $i$ with $u \in X_i$. Each node in the tree decomposition is represented
as a circle. For each set $X_i$, the diagram displays the outgoing edges from
$u$ that are contained in that node. That is, for each set $X_i$, it displays
the edges of the form $(u, w)$ with both $u \in X_i$ and $w \in X_i$.

\begin{figure}
\label{fig:simulation}
\begin{tikzpicture}[>=stealth,scale=0.65]
\node (one) at (0, 0) [node] {};
\node (two) at (5, 0) [node] {}
	edge[<-] (one);
\node (three) at (10, 0) [node] {}
	edge[<-] (two);
\node(four) at (15, 2.5) [node] {}
	edge[<-] (three);
\node(five) at (15, -2.5) [node] {}
	edge[<-] (three);

\node (fake) at (-3, 0) {}
	edge[->, dashed] (one);
\node (fake2) at (18, 2.5) {}
	edge[<-, dashed] (four);
\node (fake3) at (18, -2.5) {}
	edge[<-, dashed] (five);

\node (onel) at (0, -2.5) {$X_j$};
\node (twol) at (5, -2.5) {$X_k$};
\node (threel) at (10, -2.5) {$X_l$};
\node (fourl) at (15, 0) {$X_m$};
\node (fourl) at (15, -5) {$X_n$};

\node (v1) at (0,1)  [even] {$u$};
\node (v2) at (5,1)  [even] {$u$};
\node (v3) at (10,1)  [even] {$u$};
\node (v4) at (15,3.5)  [even] {$u$};
\node (v5) at (15,-1.5)  [even] {$u$};

\node (b) at (0.8, -0.8) [odd] {$e$}
	edge[<-] (v1);
\node (a) at (-0.8, -0.8) [even] {$d$}
	edge[<-] (v1);

\node (b) at (5.8, -0.8) [odd] {$e$}
	edge[<-] (v2);

\node (d) at (10.8, -0.8) [even] {$g$}
	edge[<-] (v3);
\node (c) at (9.2, -0.8) [odd] {f}
	edge[<-] (v3);

\node (e) at (15, 1.7) [even] {$h$}
	edge[<-] (v4);

\node (f) at (15, -3.3) [odd] {$i$}
	edge[<-] (v5);

\end{tikzpicture}
\caption{The outgoing edges from the vertex $u$ shown in a tree decomposition.}
\label{fig:subtree}
\end{figure}
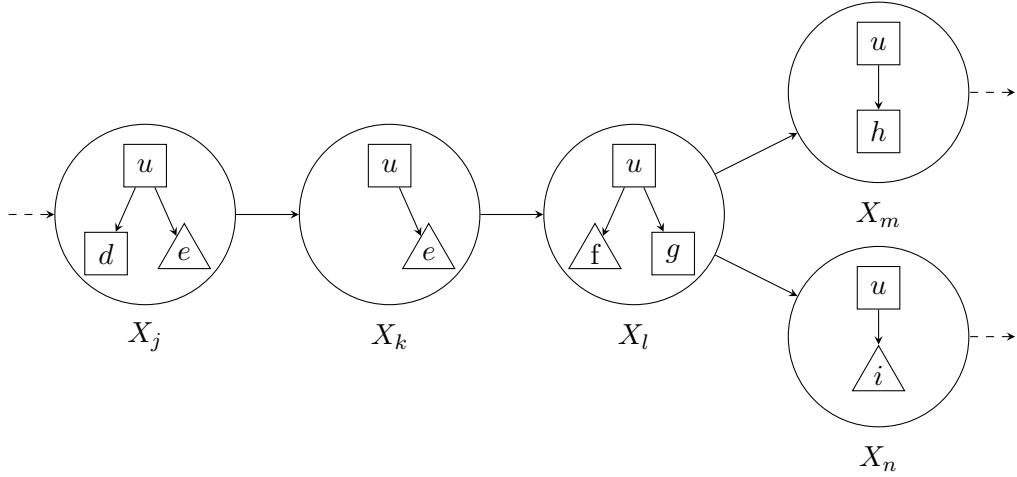

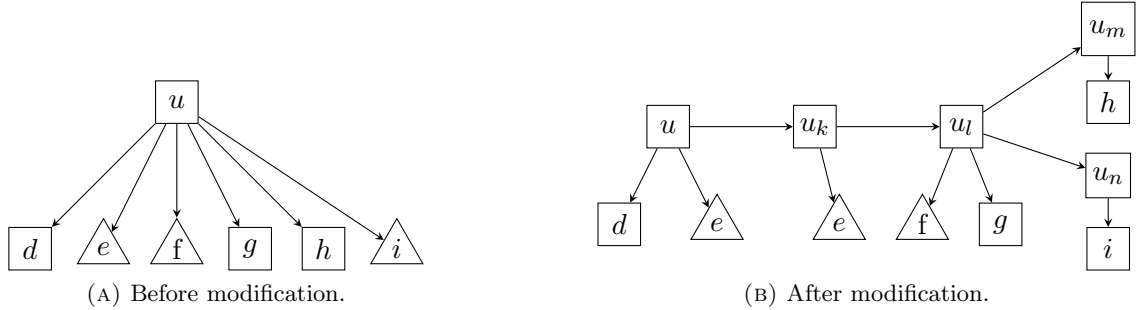
\begin{figure}
\subfloat[Before modification.]
{
	\label{fig:before}
	\begin{tikzpicture}[>=stealth,scale=0.65]
	\node (v) at (0, 0) [even] {$u$};
	\node (a) at (-3, -3) [even] {$d$}
		edge[<-] (v);
	\node (b) at (-1.5, -3) [odd] {$e$}
		edge[<-] (v);
	\node (c) at (0, -3) [odd] {f}
		edge[<-] (v);
	\node (d) at (1.5, -3) [even] {$g$}
		edge[<-] (v);
	\node (e) at (3, -3) [even] {$h$}
		edge[<-] (v);
	\node (f) at (4.5, -3) [odd] {$i$}
		edge[<-] (v);
	\end{tikzpicture}
}
\hfill
\subfloat[After modification.]
{
	\label{fig:after}
	\begin{tikzpicture}[>=stealth,scale=0.65]
	\node (v) at (0, 0) [even] {$u$};
	\node (a) at (-1, -2) [even] {$d$}
		edge[<-] (v);
	\node (b) at (1, -2) [odd] {$e$}
		edge[<-] (v);

	\node (v2) at (3, 0) [even] {$u_k$}
		edge[<-] (v);
	\node (b) at (3.5, -2) [odd] {$e$}
		edge[<-] (v2);

	\node (v3) at (6, 0) [even] {$u_l$}
		edge[<-] (v2);
	\node (c) at (5.2, -2) [odd] {f}
		edge[<-] (v3);
	\node (d) at (6.8, -2) [even] {$g$}
		edge[<-] (v3);

	\node (v4) at (9, 2) [even] {$u_m$}
		edge[<-] (v3);
	\node (e) at (9, 0.5) [even] {$h$}
		edge[<-] (v4);

	\node (v5) at (9, -1) [even] {$u_n$}
		edge[<-] (v3);
	\node (f) at (9, -2.5) [even] {$i$}
		edge[<-] (v5);

	\end{tikzpicture}
}
\caption{The modification procedure.}
\label{fig:beforeafter}
\end{figure}

Suppose that we are playing a simulation game on the set $X_j$, and that the
current vertex $c = u$. In Step~\ref{lbl:game1}, Even could select the edge $(u,
h)$. Since~$X_m$ is far from~$X_j$ in the tree decomposition, the best we can do
is to use $\log |V|$ bits to store the vertex $h$.

We will fix this problem by modifying the parity game. The modification
procedure is illustrated in Figure~\ref{fig:beforeafter}. Each vertex in the
graph will have its outgoing edges modified. Figure~\textsc{\ref{fig:before}}
shows an example vertex $u$, and Figure~\ref{fig:subtree} gives an example tree
decomposition, where only the nodes that contain $u$ are shown.
Figures~\textsc{\ref{fig:before}} shows the outgoing edges from $u$ in the
original parity game, and Figure~\textsc{\ref{fig:after}} shows the changes that
we make to the outgoing edges from $u$.

Suppose that we are playing a simulation game on $X_j$, and that we have arrived
at $u$. Now suppose that Even wants to move to some vertex not in $X_j$, for
example, the vertex $g$. The modification forces Even to a path in the tree
decomposition to a node that contains $h$, before moving to $h$.

Without this modification, we would be forced to use
$\log v$ bits to hold the variable $c$, because in Step~\ref{lbl:game1} of the
simulation game, Even could potentially pick any vertex in the graph. On the
other hand, suppose that we have made the modification, and that are playing a
simulation game on $X_j$ and that $c = u$. Let $(u, v)$ be the edge picked by
Even in Step~\ref{lbl:game1}.
\begin{iteMize}{$\bullet$}
\item If $v \in X_j$, then we can represent $v$ using $\log k$ bits.
\item If $v \notin X_j$, then we must have that $v = u_i$ for some 
tree decomposition edge $(j, i)$. In this case, we can store $v$ by remembering
$u$ and $(j, i)$.
\end{iteMize} 
In the second case, the vertex $u$ can be stored using $\log k$ bits. We claim
that the edge $(j, i)$ can be stored without using any additional memory. Recall
that we are already required to remember a tree decomposition edge of the form
$(l, j)$, in order to represent the current set $S$ and the set $F$ from the
history. Thus, to implement the simulation game on an alternating Turing
machine, we need to store two tree decomposition edges, and these edges form a
path of length two in the tree decomposition. Our overall space complexity
depends on the following simple lemma.

\begin{lem}
\label{lem:twopath}
Let $T = (I, J)$ be a rooted directed tree, with all edges oriented away from
the root. There are at most $|I|$ directed two-edge paths in $T$.
\end{lem}
\begin{proof}
Let $i \in I$ be a node. If $i$ is the root, or a successor of the root, then
there are no two-edge paths ending at $i$. Otherwise, $i$ is the is the end of
exactly one two-edge path. Therefore, there are at most $|I|$ two-edge paths in
$T$.
\end{proof}

Lemma~\ref{lem:twopath} shows that a directed two-edge path in the tree
decomposition can be stored in $\log |V|$ bits. This fact then implies the
required space bounds for the alternating Turing machine.

The rest of this section proceeds as follows. We begin by formalising the
modified parity game. We then show how the simulation game can be used to solve
the modified parity game. To do this, we introduce two 
functions $\twnext$ and $\twhistory$, and show that, if the simulation game is
equipped with $\twnext$ and $\twhistory$, then it always terminates on the
modified parity game. Finally, we show that this version of the simulation game
can be solved on an alternating Turing machine, and from this we obtain the main
result of this section.



\subsection{A modified parity game.}

For the rest of this section, we will fix a parity game $G = (V, V_\text{0},
V_\text{1}, E, \pri)$, and a rooted tree decomposition $(T = (I, J), X)$. In
this subsection, we describe a modified parity game $G'$, which will be critical
in obtaining our result. For every vertex $v \in V$, we define $\first(v, T)$ to
give the node $i \in I$ such that $v \in X_i$, and $i$ is closest to the root
in~$T$.

We can now define our modified parity game. This definition follows the idea
laid out in Figure~\ref{fig:beforeafter}. The first step is, for each vertex $v
\in V$, to add extra vertices $v_i$ for every tree decomposition node $i \in I$
other than $\first(v, T)$. Then, we add edges according to the schema laid out
by Figure~\ref{fig:beforeafter}.

\begin{defi}
\label{def:modg}
We define $G' = (V', V'_\text{0}, V'_\text{1}, E', \pri')$ as follows:
\begin{iteMize}{$\bullet$}
\item $V' = V \cup \{v_i \; : \; v \in X_i \text{ for some } i \in I \text{ and } i \ne \first(v, T)\}$.
\item $V'_\text{0} = V_\text{0} \cup \{v_i \; : \; v \in V_0 \text{ and } v_i
\in V' \text{ for some } i \in I\}$ 
\item
$V'_\text{1} = V_\text{1} \cup \{v_i \; : \; v \in V_1 \text{ and } v_i
\in V' \text{ for some } i \in I\}$.
\item The set $E'$ is constructed by the following procedures:
\begin{iteMize}{$-$}
\item For every pair of vertices $v_i, v_j \in V'$, if there is an edge $(i, j)
\in J$, then we add the edge $(v_i, v_j)$.
\item For every $v_i \in V'$, if there is a vertex $u \in X_i$ and an edge $(v,
u) \in E$, then we add $(v_i, u)$ to $E'$.
\item For every vertex $v \in V$, if there is a vertex $u \in X_{\first(v, T)}$,
and an edge $(v, u) \in E$, then we add $(v, u)$ to $E'$.
\item For every vertex $v \in V$, if there exists a vertex $v_i \in V'$ with
$(\first(v, T), i) \in J$, then we add $(v, v_i)$ to $E'$.
\end{iteMize}
\item We set $\pri'(v) = \pri(v)$ for every $v \in V$, and for every $v_i \in
V'$ we set $\pri'(v_i) = \pri(v)$.
\end{iteMize}
\end{defi}

\noindent It is not difficult to see that this construction cannot change the
winner of the parity game. This is because, if there is an edge $(v, u) \in E$,
then there is a path $\langle v_1, v_2, \dots, v_k \rangle$ in $G'$ with $v_1 =
v$, $v_k = u$, and where $v_2$ through $v_{k-1}$ belong to the owner of $v_1$.
Thus, if Even has a winning strategy in $G$ that uses the edge $(v, u)$, then we
can construct a strategy for Even in $G'$ that follows the corresponding path
from $v$ to $u$. The same property also holds for Odd's winning strategies.
Thus, we have the following lemma.

\begin{lem}
Let $s \in V$ be a vertex. Even wins from $s$ in $G$ if, and only if, Even wins from $s$ in $G'$.
\end{lem}



For the rest of this section, we will develop a version of the simulation game
that can be used to solve $G'$. Although we have modified the parity game, we
will not compute a new tree decomposition. We will develop an algorithm that
uses the original tree decomposition and the modified parity game.

\subsection{The functions \texorpdfstring{$\twnext$}{NextTW} and 
\texorpdfstring{$\twhistory$}{HistTW}.}

We now define two functions, $\twnext$ and $\twhistory$. The idea is to use the
same approach as we did in  Section~\ref{sec:time}. That is, if we want to
determine the winner for some vertex $s$ in the parity game, then we find an $i
\in I$ with $s \in X_i$, and we play $\simulate_{G'}(S, \emptyset, s, \twnext,
\twhistory)$.

We start by defining the function $\twnext$. This function will always choose
sets of the form $X_i \cup \{v_i\}$, where $i \in I$ is a tree decomposition
node, and $v_i \in V'$ is a vertex. Suppose that we are required to compute
$\twnext(S, v, \mathcal{F})$, for some set $S$, some vertex $v \in V'$, and some
history $\mathcal{F}$. Since $S$ was either chosen by $\twnext$, or was the set
chosen at the start of the game, we know that $S \subseteq X_i \cup \{v_i\}$ for
some $i \in I$ and some $v_i \in V'$. Let $(u, w)$ be an edge with $u \in S$ and
$w \notin S$. The definition of $G'$ implies that $w$ is one of the extra
vertices added that were added while constructing $G'$. That is, we have $w =
u_j$ for some tree
decomposition edge $(i, j) \in J$. Thus, we know that, if we are required to
compute $\twnext(S, v, \mathcal{F})$, then we must have $v = u_j$ for some $u
\in S$, and some tree decomposition edge $(i, j)$. We define $\twnext(S, u_j,
\mathcal{F}) = X_j \cup \{u_j\}$. 

Next, we define the $\twhistory$ function. It is very similar to the $\history$
function defined in Section~\ref{sec:time}, and it is defined as follows:
\begin{iteMize}{$\bullet$}
\item The first time that $\twhistory(\mathcal{F})$ is called, $\mathcal{F}$
will contain only one record. We do not delete this record, so we set
$\history(\mathcal{F}) = \mathcal{F}$.
\item Every other time that $\history(\mathcal{F})$ is called, $\mathcal{F}$
will contain exactly two records. Due to the definition of $\twnext$, we know
that the two records in $\mathcal{F}$ have the form $(S_1, v_1,
P_1)$ and $(S_2, v_2, P_2)$, where $X_i \subseteq S_1$ and $X_j \subseteq S_2$, for some $i, j \in I$. 
Furthermore, we know that $(i, j) \in J$ is a directed edge in the tree
decomposition. We set $\history(\mathcal{F}) = \{(S_2, v_2, P_2)\}$.
\end{iteMize}

\noindent The next lemma shows that, if we equip the simulation game with the
functions $\twnext$ and $\twhistory$, then the simulation game always
terminates. Since we are following the approach used in Section~\ref{sec:time},
this proof is similar in structure to the proof of Lemma~\ref{lem:terminates}.

\begin{lem}
\label{lem:twterminates}
If $i$ is the root of the tree decomposition, and $s \in X_i$ is a vertex, then
$\simulate_{G'}(X_i, \emptyset, s, \twnext, \twhistory)$ terminates.
\end{lem}
\begin{proof}
We begin by arguing that Odd cannot play \rej\ an infinite number of times. This
follows from the definition of $\twnext$: we know that if Odd plays \rej\ in a
simulation game where $X_i \subseteq S$ for some $i \in I$, then the next game
will be played on a set $S'$ with $X_j \subseteq S'$ for some directed edge $(i,
j) \in J$. Thus, after Odd has played \rej\ at most $|I|$ times, we will have
arrived at a leaf in the tree decomposition.

The $\twhistory$ function ensures that, once we reach a leaf in the tree
decomposition, the game must end. Suppose that we have reached a simulation
game on $S$, where $S = X_j \cup \{w_j\}$ where $j \in I$ is a leaf, and $w \in
X_j$ is a vertex. The history function ensures
that $\mathcal{F}$ contains a record $(F, p, P)$ with $X_k \subseteq F$,
where $(k, j) \in J$. We now have two properties:
\begin{iteMize}{$\bullet$}
\item For every edge $(v, u)$ with $v \in X_j \setminus X_k$ and $u \notin X_j$,
we have $u \in X_k$. This follows from the definition of a tree decomposition
given in Definition~\ref{def:treedecomp}.
\item For every edge $(v, u)$ with $v = w_j$, we have $u \in X_j$. This follows
from the definition of $E'$ given in Definition~\ref{def:modg}.
\end{iteMize}
Thus, we have shown that, for every edge $(v, u)$ with $v \in X_j \setminus X_k$ and $u \notin
X_j$, we have $u \in X_k$. Therefore, since $X_k \subseteq F$,
Step~\ref{lbl:game2} of the simulation game ensures that Odd cannot play \rej.

Since $S$ contains at most $k+1$ vertices, we know that, if Odd can no longer
play \rej, then Odd can play \acc\ at most $k + 2$ times before the game is
forced to end in either Step~\ref{lbl:game5} or Step~\ref{lbl:game6}.
\end{proof}

Now that we have established termination, we can apply
Lemma~\ref{lem:correct} to imply the correctness
of this version of the simulation game.

\subsection{Implementation on an alternating Turing machine}

Finally, we show how this version of the simulation game can be implemented on
an alternating Turing machine. As usual, the existential and universal states of
the alternating Turing machine will be used to simulate the moves of Even and
Odd. We have the following theorem.

\begin{thm}
Given a parity game~$G$, a tree decomposition $(T = (I, J), X)$ of width~$k$ for
$G$, and a vertex $s$, we can determine the winner of $s$ in $O(|V| \cdot
(k+1)^{k+5} \cdot (|D|+1)^{3k + 5})$ time.
\end{thm}
\begin{proof}
For the most part, the proof of this theorem is the same as the proof of
Theorem~\ref{thm:time}. Recall, from the proof of Theorem~\ref{thm:time}, that
we must remember the following variables: $S$, $\mathcal{F} = \{F, p, P\}$,
$\Pi$, $P'$, $c$, and $v$.

Our central claim is that $S$, $F$, and $v$ can be stored in $\log |V| + 2 \cdot
\log (k+1)$ bits. Note that, $\twnext$ and $\twhistory$ ensure that $S \subseteq X_i
\cup \{s_i\}$ and $F \subseteq X_j \cup \{f_j\}$, for some $(j, i) \in J$, and
some vertices $s, f \in V$. By Definition~\ref{def:modg}, there cannot be an
edge $(v, u)$ with $v \in S \setminus F$, and $u = f_j$. Thus, our alternating Turing
machine can represent $F$ as the set $X_j$, while forgetting the vertex $f_j$,
and this cannot affect the correctness of the algorithm.

Now we argue that $S$, $F$, and $v$ can be stored in $\log |V| + 2 \log (k+1)$
bits. There are two cases to consider:
\begin{iteMize}{$\bullet$}
\item If $v \in S$, then $X_i$ and $X_j$ can be represented by pointing to an
edge in $J$. Since $T$ is a tree, we know that $|J| \le |I| \le |V|$. Then, to
store $s_i$ we can simply store the vertex $s$ in $\log k$ bits, because we already know $X_i$.
Since $v \in S$, and since $|S| \le k + 1$, we can store $v$ using $\log (k+1)$
bits. 
\item If $v \notin S$, then we know that $v = u_l$ for some $u \in S$, and some
$l \in I$. Note that $j$, $i$, and $l$ form a directed two-edge path in $T$.
Thus, by Lemma~\ref{lem:twopath}, we can represent $j$, $i$, and $l$, in $\log
|I| \le \log |V|$ bits. As with the previous case, we can represent the extra
vertex $s_i$ using $\log k$ bits. Moreover, since $u \in S$, we can represent
$v$ using at most $\log (k+1)$ bits.
\end{iteMize}
Therefore, in both cases we can store $S$, $F$, and $v$ using $\log |V| + \log
(k+1) + \log k$ bits.

The amount of space for the variables $p$, $P$, $\Pi$, $P'$, and $c$ follow from
the arguments made in the proof of Theorem~\ref{thm:time}. However, there is one
small difference: in this proof, we have $|S| = k+1$ rather than $|S| = k$.
Thus, while $P$ could be stored in $k \cdot \log (D+1)$ bits in
Theorem~\ref{thm:time}, we must use $(k+1) \cdot \log (D+1)$ bits in this proof.
Similar modifications are required for the other variables.

In total, therefore, we require at most $\log |V| + (k+5) \cdot \log(k+1) + (3k
+ 5) \cdot \log (|D| + 1)$ bits of storage in our alternating Turing machine.
This alternating Turing machine can be implemented on a deterministic Turing
machine in $O(|V| \cdot (k+1)^{k+5} \cdot (|D|+1)^{3k + 5})$ time.
\end{proof}

\section{Parity Games With Bounded Treewidth Are In 
\texorpdfstring{NC$^2$}{NC2}}\label{sec:space}

We begin by defining the $\nextfun$ and $\history$ functions that will be used
in by simulation game in this result. Recall, from Section~\ref{sec:outline2},
that this version of the simulation game gives extra strategic choices to player
Odd. We use this idea to define $\ncnext$ and $\nchistory$ as follows:
\begin{iteMize}{$\bullet$}
\item When Step~\ref{lbl:game4} requires a value for $\ncnext(S, v,
\mathcal{F'''})$, we allow player Odd to select the set of vertices that will be
used. More precisely, we allow Odd to select $S' \subseteq V$ such that $|S'|
\le k$ and $v \in S'$. 
\item When Step~\ref{lbl:game4} requires a value for
$\nchistory(\mathcal{F''})$, we allow Odd to to select the set $\mathcal{F'''}$
by deleting records from $\mathcal{F''}$. We insist that, if $|\mathcal{F''}| >
3$, then Odd must delete enough records so that $|\mathcal{F}| \le 3$.
\end{iteMize}
\noindent Thus, we will use 
$\simulate(S, \mathcal{F}, s, \ncnext, \nchistory)$ as the simulation game in
this section.

Note that this simulation game could go on forever, since Odd could, for
example, always play~$\rej$, and remove all records from $\mathcal{F}$. For this
reason, we define a limited-round version of the game, which will be denoted by
$\simulate^r(S, \mathcal{F}, s, \ncnext, \nchistory)$, where $r \in \nats$. This
game is identical to $\simulate^r(S, \mathcal{F}, s, \ncnext, \nchistory)$ for
the first~$r$ rounds. However, if the game has not ended before the start of
round $r+1$, then Even is declared to be the winner.

Let $r \in \nats$ be an arbitrarily chosen bound. The following lemma states
that, no matter what sets Odd chooses, if the simulation game stops before round
$r+1$, then the winner of the parity game is correctly determined. 

\begin{lem}
\label{lem:winspace}
No matter which sets are chosen by Odd for $\ncnext$ and $\nchistory$ in
Step~\ref{lbl:game4}:
\begin{iteMize}{$\bullet$}
\item if $s \in W_0$,
then Even has a strategy to win $\simulate^r(\{s\}, \emptyset, s, \ncnext,
\nchistory)$. 
\item If $s \in W_1$, and if the game ends before round $r+1$, then Odd
has a strategy to win $\simulate^r(\{s\}, \emptyset, s, \ncnext, \nchistory)$.
\end{iteMize}
\end{lem}

\noindent The correctness of this lemma follows from Lemma~\ref{lem:correct}.
Note that this lemma does not guarantee anything if the simulation game
terminates in round $r+1$. Our goal in this section is to show that there is an
$r$ such that Odd can always force the game to end before round $r+1$, which
will then allow us to apply Lemma~\ref{lem:winspace} to prove correctness of
our algorithm.

We will assume that our parity game has a tree decomposition $(T = (I, J), X)$
of width $k$. Our main goal is to show that Odd has a strategy for choosing
$\ncnext$ and $\nchistory$ that forces the game to terminate within the
first~$r$ rounds, for some $r \in \nats$.  In particular, to show that parity
games with bounded treewidth lie in NC$^2$, our strategy must have the following
properties:
\begin{iteMize}{$\bullet$}
\item The game must end in $O(k \cdot \log |V|)$ rounds. That is, we have $r = c
\cdot k \cdot \log |V|$ for some constant $c$.
\item The history may never contain more than three records.
\end{iteMize}
\noindent In the following subsections, we will provide such a strategy.

\subsection{Odd's strategy for \texorpdfstring{$\ncnext$}{NextNC} and \texorpdfstring{$\nchistory$}{HistNC}.}

Recall, from Section~\ref{sec:outline2}, that the strategy for Odd consists of
two sub-strategies: $\slice$ and $\reduce$. The strategy $\slice$ will be
applied whenever $\mathcal{F}$ contains fewer than~$3$ records, and $\reduce$
will be applied whenever $\mathcal{F}$ contains exactly~$3$ records. Both of
these strategies maintain a set of \emph{eligible} tree decomposition nodes $L
\subseteq I$. The set of eligible tree decomposition nodes represents the
portion of the parity game that can still be reached by the simulation game. An
intuitive idea of this is given in Figures~\textsc{\ref{fig:slice}}
and~\textsc{\ref{fig:reduce}}.
In both cases, the set of eligible nodes is the subtree between the nodes $F_1$,
$F_2$, and $F_3$. 

We give a formal definition for how $L$ will be updated. Suppose that Odd
rejects the profile~$P$ for the vertex~$v$ and set $X_i$. We define
$\subtree(X_i, v)$ to be the set of nodes in the subtree of~$i$ rooted at
$\direction(X_i, v)$. We update $L$ as follows.
\begin{iteMize}{$\bullet$}
\item The first time that \rej\ is played, Odd sets $L = I$.
\item In each subsequent iteration, the set~$L$ is updated so that $L := L \cap
\subtree(X_i, v)$.
\end{iteMize}

\noindent Odd uses $L$ to decide which records should be deleted. Formally, Odd will
select $\nchistory(\mathcal{F})$ according to the following rules:
\begin{iteMize}{$-$}
\item The first time that \rej\ is played, Odd removes all records from
$\mathcal{F}$.
\item In each subsequent iteration, Odd uses the following rule. Let $(F, p, P)
\in \mathcal{F}$, and suppose that $X_i \subseteq F$ for some $i \in I$. If
there is no edge $(i, j) \in J$ with $j \in L$, then Odd removes $(F, p, P)$
from $\mathcal{F'}$.
\end{iteMize}

We now define the strategy~$\slice$. Recall that
Lemma~\ref{lem:splittree} shows that every tree with more than three nodes can
be split into pieces, where each piece has at least two-thirds of the nodes in
the tree. We define $\splitgame(I, J)$ to be a function that returns a vertex $i
\in I$ that satisfies the properties given in Lemma~\ref{lem:splittree}.  The
strategy $\slice$ will use the function $\splitgame$ to remove at least one
third of the eligible nodes. Formally, we define $\slice(S, v, L, \mathcal{F})$.
The strategy is defined as follows: 
\begin{iteMize}{$\bullet$}
\item If $|L| \ge 3$ then $\slice(S, v, L, \mathcal{F}) = X_i$ where $i = \splitgame(L, J)$. 
\item If $|L| < 3$ then $\slice(S, v, L, \mathcal{F}) = X_i$ for some
arbitrarily chosen $i \in L$.
\end{iteMize}
The second clause is necessary, because Lemma~\ref{lem:splittree} requires that
the tree should have at least three vertices.

We now introduce our second strategy for choosing $S'$, which is called
$\reduce$. This strategy will be applied whenever $|\mathcal{F}| = 3$. Recall
that the purpose of~$\reduce$ is to reduce the number of records in
$\mathcal{F}$ to at most $2$, so that we can continue to apply~$\slice$. Suppose
that $\mathcal{F} = \{(F_1, p_1, P_1), (F_2, p_2, P_2), (F_3, p_3, P_3)\}$. By
assumption we have that $F_1 = X_i$, $F_2 = X_j$, and $F_3 = X_k$, for some
$i,j,k \in I$. 
Note that $i$, $j$ and
$k$ cannot all lie on the same path, because otherwise one of the records in
$\mathcal{F}$ would have been deleted by $\nchistory$.
 It is a basic property of trees that, if $i$, $j$, and $k$ are not on the same
path, then there exists a unique
vertex $l \in I$ such that, if~$l$ is removed, then $i$, $j$, and $k$, will be
pairwise disconnected.
For each history $\mathcal{F}$ with $|\mathcal{F}| = 3$,
we define $\point(\mathcal{F})$ to be the function that gives the vertex~$l$. We
define the strategy $\reduce(S, v, L, \mathcal{F}) = X_i$, where $i =
\point(\mathcal{F})$. 

Finally, we can give the full strategy for Odd.  Odd selects $\ncnext(S, v,
\mathcal{F})$ according to the following rules:
\begin{iteMize}{$\bullet$}
\item If $|\mathcal{F}| < 3$, then Odd selects $S' = \slice(S, v, L,
\mathcal{F})$.
\item If $|\mathcal{F}| = 3$, then Odd selects $S' = \reduce(S, v, L,
\mathcal{F})$.
\end{iteMize}

\subsection{Correctness of \texorpdfstring{$\sr$}{SR}}

In this subsection, we show that $\sr$ has both of the properties that we
desire. Suppose that Odd uses $\sr$, and that the outcome is a play of the
simulation game:
\begin{equation*}
\langle \mathcal{S}_1 = \simulate(S_1, \mathcal{F}_1, s_1, \ncnext, \nchistory),
\mathcal{S}_2 = \simulate(S_2, \mathcal{F}_2, s_2, \ncnext, \nchistory),
\dots \rangle.
\end{equation*}
Furthermore, let $L_i$ be the contents of $L$ at the start of the game
$\mathcal{S}_i$, and let $v_i = s_{i+1}$ be the final vertex visited in
$\mathcal{S}_i$. The first property that we will prove is that, as long as Odd
plays $\sr$, the set $L$ is guarded by the records in $\mathcal{F}$. Informally,
this means that, in order to move from a tree decomposition node $i \in L$, to a
tree decomposition node $j \notin L$, we must pass through some record in
$\mathcal{F}$. Formally, this can be expressed as the following lemma.

\begin{lem}
Suppose that Odd plays $\sr$.
If $z \ge 2$, then for every edge $(i, j) \in J$ with $i \in L_z$, and $j \notin
L_z$, there exists
a record $(F, p, P) \in \mathcal{F}_{z}$ with $X_z = F$.
\end{lem}
\begin{proof}
We will prove this claim by induction on $z$. The base case is vacuously true,
since $\sr$ sets $L_2 = I$, which implies that there cannot exist an edge $(i,
j)$ with $j \notin L_2$. 

For the inductive step, assume that this lemma holds for some $z > 2$.
Since $L_{z+1} = L_{z} \cap
\subtree(S_{z}, v_{z})$, the only possible edge that could be not covered by
a set in $\mathcal{F}_z$ is the edge $(i, j)$, where $X_i = S_{z+1}$, and $j \in
\subtree(S_{z+1}, v_{z+1})$. However, this edge is covered by two properties:
\begin{iteMize}{$\bullet$}
\item In Step~\ref{lbl:game4}, the simulation game always adds the record $(F,
-, P')$ to $\mathcal{F'}$ to create $\mathcal{F''}$. 
\item When choosing $\mathcal{F'''}$, the strategy $\sr$ never removes a record
$(F, p, P)$ from $\mathcal{F''}$ with $F \subseteq X_j$, where $(i, j) \in J$ and $i \in
L_{z+1}$ and $j \notin L_{z+1}$.
\end{iteMize} 
\noindent Thus, a record $(F, p, P)$ with $F = X_i$ must be contained in
$\mathcal{F}_{z+1}$, which completes the proof of the inductive step.
\end{proof}

Next we will show the correctness of $\reduce$. In particular, we show how if it
is used when $|\mathcal{F}_x| = 3$, it enforces that $|\mathcal{F}_{x+1}| \le
2$.

\begin{lem}
Suppose that Odd plays $\sr$, and for some $x \ge 2$ we have $|\mathcal{F}_{x}|
= 3$, and $S_{x} = \reduce(S_{x-1}, v_{x-1}, L_{x-1}, \mathcal{F}_x)$. We have
$|\mathcal{F}_{x+1}| \le 2$.
\end{lem}
\begin{proof}
Since $x \ge 2$, and since Odd is following $\sr$, we know that there must exist
three tree decomposition nodes $i, j, k \in I$, such that, if $(F_1, p_1, P_1)$,
$(F_2, p_2, P_2)$, and $(F_3, p_3, P_3)$ are the three records in
$\mathcal{F}_x$, then we have $F_1 = X_1$, $F_2 = X_j$, and $F_3 = X_k$. 
By definition, $\reduce(S_{x-1}, v_{x-1}, L_{x-1})$ selects a vertex~$l$ such
that, if~$l$ is removed from $T = (I, J)$, then $i$, $j$, and $k$ become
disconnected. Thus, at most one of $i$, $j$, and $k$ can lie in
$\subtree(S_{x-1}, v_{x-1})$, and hence at least two records will be removed by
$\sr$. Therefore, since at most one new record is added to $\mathcal{F}_{x+1}$,
we must have $|\mathcal{F}_{x+1}| \le |\mathcal{F}_x| + 1 - 2 = 2$. Thus, we
have shown that $|\mathcal{F}_{x+1}| \le 2$.
\end{proof}

Finally, we are able to prove that $\sr$ has the property that we
require.

\begin{lem}
\label{lem:oddchooseset}
Suppose that the parity game has a tree decomposition of width $k$. If Odd plays
$\sr$, then the game ends in $O(k \cdot \log |V|)$ rounds. 
\end{lem}
\begin{proof}
Lemma~\ref{lem:splittree} implies that, if  $\slice$ is invoked at the end of
$\mathcal{S}_i$, then we will have $|L_{i+1}| \le \frac{2}{3}|L_{i}|$. Each time $\slice$
is used, we add at most one new record to $\mathcal{F}$. Therefore, for each use
of $\slice$, we must use $\reduce$ at most~$1$ time, in order to keep
$|\mathcal{F}| \le 3$. Thus, after Odd has played \rej\ at most $2 \cdot
\log_{\frac{3}{2}}(|I|)$ times, we will arrive at a game $\mathcal{S}_c$ with $|L_c| < 3$.

We claim that once we have $|L_c| < 3$, Odd can play \rej\ at most two more
times. This follows from the fact that for every edge $(i, j) \in J$ with $i \in
L$ and $j \notin L$, we must have $X_j \subseteq F_a$ for some~$F_a$ in
$\mathcal{F}_c$.
Therefore, after Odd plays \rej\ two more times, we will arrive at a game
$\mathcal{S}_{c+2} = \simulate_G(S, \mathcal{F}, s, \ncnext, \nchistory)$ such
that, if $X_i \subseteq S$, then for all edges $(i, j) \in J$ we have $X_j = F_a$ for
some set~$F_a$ in
$\mathcal{F}$. This is because the vertex $i \in I$ such that $S = X_i$ is
always removed from $L$, and so it does not matter whether Odd uses $\slice$ or
$\reduce$ at this point. In $\mathcal{S}_{c+2}$, we know that Step~\ref{lbl:game2}
prevents Odd from playing \rej, and therefore the game will end.

So far, we have that Odd can play \rej\ at most $2 \cdot \log_{\frac{3}{2}}(|I|)
+ 2$ times. Note that between each instance of \rej, Odd can play \acc\ at most
$k + 1$ times without triggering Step~\ref{lbl:game5} or Step~\ref{lbl:game6}.
Since $|I| \le |V|$ we can therefore conclude that the game can last at most
$(k+1) \cdot (2 \cdot \log_{\frac{3}{2}}(|V|) + 2)$ times, and this is contained
in $O(k \cdot \log |V|)$.
\end{proof}

\subsection{Implementation on an alternating Turing machine}

Let~$G$ be a parity game, and let $G' = (V, V'_\text{0}, V'_\text{1}, E, \pri')$
be a parity game that swaps the roles of the two players in~$G$. More formally
we have:
\begin{iteMize}{$\bullet$}
\item $V'_0 = V_1$, and $V'_1 = V_0$.
\item For each $v \in V$ we set $\pri'(v) = \pri(v) + 1$.
\end{iteMize}
It should be clear that Even can win from a vertex $s \in V$ in $G$ if and only
if Odd can win from~$s$ in~$G'$.

Suppose that we want to find the winner of a vertex~$s$. We choose $c = O(k \cdot \log |V|)$, where
the constant hidden by the $O(\cdot)$ notation is given in the proof of
Lemma~\ref{lem:oddchooseset}. We then solve $\simulate^c_G(\{s\}, \emptyset,
s, \ncnext, \nchistory)$ and $\simulate^c_{G'}(\{s\}, \emptyset, s, \ncnext,
\nchistory)$. If both games declare that Even
wins $s$, then we declare that the game has treewidth larger than~$k$.
Otherwise, we declare that the winner of 
$\simulate^c_G(\{s\}, \emptyset, s, \ncnext, \nchistory)$ is the winner of $s$.
Lemmas~\ref{lem:evenwinspace},~\ref{lem:oddwinspace}, and~\ref{lem:oddchooseset}
imply that this procedure is correct.

Thus, to prove our main result, we must show that $\simulate^c_G(\{s\},
\emptyset, s, \ncnext, \nchistory)$ can be solved in $O(k^2 \cdot (\log |V|)^2)$
time and $O(k \cdot \log |V|)$ space by an alternating Turing machine. To
implement the game on an alternating Turing machine, we use the existential
states to simulate Even's moves and the universal states to simulate Odd's
moves. We begin by proving the time bound.

\begin{lem}
\label{lem:atime}
A simulation of $\simulate^c_G(\{s\}, \emptyset, s, \ncnext, \nchistory)$ by
an alternating Turing machine uses at most $O(k^2 \cdot (\log |V|)^2)$ time.
\end{lem}
\begin{proof}
By definition we have that $\simulate^c_G(\{s\}, \emptyset, s, \ncnext,
\nchistory)$ allows the
game to continue for at most $O(k \cdot \log |V|)$ rounds. Each step of the game
takes the following time:
\begin{iteMize}{$\bullet$}
\item Since we are only required to guess a vertex in Step~\ref{lbl:game1},
this step can be implemented in $\log |V|$ time.

\item In Step~\ref{lbl:game2} checking whether $v \in S$ can be done in $k
\cdot \log |V|$ time. Note that we can have at most $3$ records in
$\mathcal{F}$. Therefore, the check for whether $v \in F$ for some $(F, p, P)
\in \mathcal{F}$ can also be performed in $3 \cdot k \cdot \log |V|$ time.

\item Generating a strategy profile in Step~\ref{lbl:game3} can take at most $k
\cdot \log(|D| + 1)$ time.

\item If Odd plays \acc\ in Step~\ref{lbl:game4}, then, if we always maintain a
pointer to the end of $\Pi$, appending $(v, P(u), u)$ to~$\Pi$ will take $2
\cdot \log |V| + \log |D|$ time.

\item If Odd plays \rej\ in Step~\ref{lbl:game4} then:
\begin{iteMize}{$-$}
\item Since computing $\maxpri(\Pi)$ can take at most $k \cdot \log |D|$ time,
and since $|\mathcal{F}| \le 3$, we have that $\mathcal{F'}$ can computed in $k
\cdot \log |D|$ time.
\item Appending $(F, -, P')$ to $\mathcal{F'}$ when creating $\mathcal{F''}$
will take $k \cdot \log |V| + \log(|D|+1) + k \cdot \log(|D| + 1)$ time.
\item Implementing Odd's strategy for $\nchistory$ can be done in $4 \cdot k
\cdot \log |V| + \log (|D| + 1) + k \cdot \log(|D| + 1)$ time. This is because
Odd is only allowed to remove records from $\mathcal{F''}$, and because
$|\mathcal{F''}| \le 4$.
\item Implementing Odd's strategy for $\ncnext$ involves picking a subset of
vertices with size at most $k$. This can be done in $k \cdot \log |V|$ time.
\end{iteMize}

\item In Step~\ref{lbl:game5}, since $|\mathcal{F}_i| \le 3$, we can check
whether $c \in F$, for some~$(F, p, P) \in \mathcal{F}_i$, in $3 \cdot k \cdot
\log |V|$ time. If the game stops here, then computing~$p'$, which requires us
to find $\maxpri(\Pi)$, must take at most $k \cdot \log |D|$ time.

\item In Step~\ref{lbl:game6}, determining if $\Pi$ forms a cycle can be done
by checking for each tuple $(v, p, u) \in \Pi$ if $c = v$. Since there can be at
most~$k$ tuples in $\Pi$, this can be done in $k \cdot \log |V|$ time. Finding
the highest priority on the cycle then takes at most $k \cdot \log |D|$ time.
\end{iteMize}
Since $|D| \le |V|$, we have that each round takes at most $O(k \cdot \log |V|)$
time. Since there are at most $O(k \cdot \log |V|)$ rounds, the machine must
terminate in $O(k^2 \cdot (\log |V|)^2)$ time. 
\end{proof}

Next, we prove the space bounds for the alternating Turing machine.

\begin{lem}
\label{lem:aspace}
A simulation of $\simulate^c_G(\{s\}, \emptyset, s, \ncnext, \nchistory)$ by an alternating Turing
machine uses at most $O(k \cdot \log |V|)$ space.
\end{lem}
\begin{proof}
In order to simulate the game, we must remember the set $S$, the vertex $c$ and
the vertex $v$, the strategy profile $P'$, the path $\Pi$, and the history
$\mathcal{F}$. It is important to note that we do not need to remember the set
$L$ used by $\sr$. In Lemma~\ref{lem:oddchooseset} we showed that there
\emph{exists} an Odd strategy that forces the game to terminate in $O(k \cdot
\log |V|)$. The existence of such a strategy is sufficient to ensure that the
alternating Turing machine computes the correct answer, and we do not need to
actually implement the strategy on the alternating Turing machine.

Therefore, the amount of space used by the alternating Turing machine is as
follows.
\begin{iteMize}{$\bullet$}
\item Since we always have $|S| \le k$, we know that $S$ requires at most $k
\cdot \log |V|$ bits to store. 
\item The vertices $c$ and $v$ each require $\log |V|$ bits to store. 
\item The strategy profile $P'$ contains $|S|$ mappings $u \rightarrow p$, where
$u \in S$ and $p \in D \cup \{ - \}$. Therefore, $P'$ can be stored using
$k (\log |V| + \log(|D| + 1))$ bits.
\item Since $|S| \le k$, the path $\Pi$ contains at most $k$ tuples of the form
$(v, p, u)$, where $v, u \in V$ and $p \in D \cup \{-\}$. Therefore, storing the
path~$\Pi$ requires at most $3 \cdot k \cdot (\log |V| + \log(|D| + 1))$ bits.
\item For each record $(F, p, P) \in \mathcal{F}$, the set $F$ requires $k \cdot
\log |V|$ bits to store, the priority $p$ requires $\log(|D| + 1)$ bits to
store, and we have already argued that the strategy profile~$P$ requires at most
$k \cdot (\log |V| + \log(|D| + 1))$ bits to store. Since $|\mathcal{F}| \le 3$,
we have that $\mathcal{F}$ requires at most $3 \cdot (2 \cdot k \cdot \log |V| +
(k + 1) \cdot \log(|D| + 1))$ bits to store.
\end{iteMize}
Since $|D| \le |V|$, we have shown that the alternating Turing machine requires
at most $O(k \cdot \log |V|)$ space. 
\end{proof}

Having shown these two properties, we now have the main result of this section.

\begin{thm}
\label{thm:space}
Let~$G$ be a parity game and~$k$ be a parameter. There is an alternating Turing
machine that takes $O(k^2 \cdot (\log |V|)^2)$ time, and uses $O(k \cdot \log
|V|)$ space, to either determine the winner of a vertex $s \in V$, or correctly
report that~$G$ has treewidth larger than~$k$.
\end{thm}

A simple corollary of the time and space results for our alternating Turing
machine is that our problem lies in the class NC$^2$~\cite[Theorem
22.15]{Allender}. 

\begin{cor}
The problem of solving a parity game with bounded treewidth lies in $NC^2$.
\end{cor}

\section{Conclusion}

We have seen three results: a 
$O(|V|^{k+3} \cdot k^{k + 2} \cdot (|D| + 1)^{3k + 2})$ time algorithm for
parity games with bounded DAG width, a 
$O(|V| \cdot (k+1)^{k + 5} \cdot (|D| + 1)^{3k + 5})$ time algorithm for parity
games with bounded treewidth, and a proof that parity games with bounded
treewidth lies in NC$^2$.

It is worth noting that the running time of our algorithm for parity games with
bounded DAG width includes a factor of $|V|^{k+2}$, because that is the best
known upper bound for the number of edges that can appear in a DAG
decomposition. If, in the future, better upper bounds are derived, then the
running time of our algorithm will see a corresponding improvement.

An interesting open problem is: are there \emph{fixed parameter tractable}
algorithms for solving parity games with bounded treewidth? An algorithm is
fixed parameter tractable if its running time can be expressed as
$O(\mathit{poly}(V) \cdot f(k))$, where the degree of the polynomial is
independent of the treewidth. Neither of the algorithms that we have presented
satisfy this property, because they include the terms $|V|^{k+3}$ and $(|D| +
1)^{3k + 5}$, respectively. In the case of DAG width, if the upper bounds on the
size of the DAG decomposition are tight, then it seems unlikely that there exist
fixed parameter tractable algorithms. However, for the case of treewidth, it is
entirely possible that such algorithms could exist.

\bibliographystyle{alpha}
\bibliography{references}

\end{document}